\documentclass[11pt, letterpaper]{article}
\usepackage{fullpage}
\usepackage{amsthm}
\usepackage{amsmath,amssymb,amsfonts,nicefrac}
\usepackage{xspace}
\usepackage{color}
\usepackage{url}
\usepackage{hyperref}
\usepackage{bm}
\usepackage{bbm}
\usepackage{times}

\usepackage{enumitem}

\newtheorem{thm}{Theorem}[section]

\newtheorem{lemma}[thm]{Lemma}

\newtheorem{claim}[thm]{Claim}

\newtheorem{definition}[thm]{Definition}

\newtheorem{fact}[thm]{Fact}

\newcommand\E{\mathop{\mathbb{E}}}
\newcommand\card[1]{\left| {#1} \right|}
\newcommand\sett[2]{\left\{ \left. #1 \;\right\vert #2 \right\}}

\newcommand\set[1]{{\left\{ #1 \right\}}}
\newcommand\Prob[2]{{\Pr_{#1}\left[ {#2} \right]}}

\newcommand\cProb[3]{{\Pr_{#1}\left[ \left. #3 \;\right\vert #2 \right]}}

\newcommand\Expect[2]{{\mathop{\mathbb{E}}_{#1}\left[ {#2} \right]}}

\newcommand\cExpect[3]{{\mathbb{E}_{#1}\left[ \left. #3 \;\right\vert #2 \right]}}
\newcommand\norm[1]{\left\| #1 \right\|}

\newcommand\defeq{\stackrel{def}{=}}

\newcommand\inner[2]{\langle{#1},{#2}\rangle}
\newcommand\eps{\varepsilon}

\renewcommand\geq{\geqslant}
\renewcommand\leq{\leqslant}

\newcommand\MI{\mathrm{I}}

\newcommand{\HH}{\mathrm{H}}

\newcommand{\DKL}[2]{\mathrm{D}_{\text{KL}}\left( #1 \parallel #2 \right)}

\newcommand{\rom}[1]{\uppercase\expandafter{\romannumeral #1\relax}}

\title{Rounding via Low Dimensional Embeddings}
\author{Mark Braverman
\thanks{Department of Computer Science, Princeton University.
Research supported in part by the NSF Alan T. Waterman Award, Grant No. 1933331, a Packard Fellowship in Science and Engineering, and the Simons Collaboration on Algorithms and Geometry. }
\and Dor Minzer \thanks{Department of Mathematics, Massachusetts Institute of Technology, Cambridge, USA. Supported by a Sloan Research
Fellowship.}}
\date{\vspace{-5ex}}
\begin{document}
\maketitle
\begin{abstract}
A regular graph $G = (V,E)$ is an $(\eps,\gamma)$ small-set expander if for any set of vertices of fractional size at most $\eps$,
at least $\gamma$ of the edges that are adjacent to it go outside. In this paper, we give a unified approach to several known complexity-theoretic
results on small-set expanders.
In particular, we show:
\begin{enumerate}
  \item Max-Cut: we show that if a regular graph $G = (V,E)$ is an $(\eps,\gamma)$ small-set expander that contains a cut of fractional
  size at least $1-\delta$, then one can find in $G$ a cut of fractional size at least $1-O\left(\frac{\delta}{\eps\gamma^6}\right)$ in polynomial time.

  \item Improved spectral partitioning, Cheeger's inequality and the parallel repetition theorem over small-set expanders. The general form of each one
  of these results involves square-root loss that comes from certain rounding procedure, and we show how this can be avoided over small set expanders.
\end{enumerate}
Our main idea is to project a high dimensional vector solution into a low-dimensional space while roughly maintaining $\ell_2^2$ distances, and then
perform a pre-processing step using low-dimensional geometry and the properties of $\ell_2^2$ distances over it.
This pre-processing leverages the small-set expansion property of the graph to transform a vector valued solution
to a different vector valued solution with additional structural properties, which give rise to more efficient integral-solution rounding schemes.
\end{abstract}

\section{Introduction}

Expander graphs are important combinatorial objects with myriad of applications throughout theoretical computer science (see~\cite{Expander}).
One of the reasons for that are the numerous equivalent definitions of expander graphs, that offer different views on this basic object.
Combinatorially, a $d$-regular graph $G = (V,E)$ is said to be a combinatorial expander if the {\it edge expansion} of any $S\subseteq V$ of size at most $n/2$ is
at least an absolute constant.
\begin{definition}
  Let $G = (V\cup U,E)$ be an undirected regular graph. The edge expansion of $S\subseteq V$ is defined as
  $\Phi(S) = \frac{\card{\sett{(v,u)}{v\in S, u\not\in S}}}{\card{\sett{(v,u)}{v\in S}}}$. In words,
  $\Phi(S)$ is the fraction of edges adjacent to $S$ going outside it.
\end{definition}

Spectrally, a graph is said to be an expander if, letting $A_G$ be the normalized
transition matrix of $G$, the second eigenvalue of $A_G$ is at most $1-\Omega(1)$.
A basic result in spectral graph theory, called Cheeger's inequality, asserts that
qualitatively, spectral expanders and combinatorial expanders are equivalent.

The main object studied in this paper is small set expanders, which is a relaxation of expander graphs defined as follows.
\begin{definition}
  A graph $G = (V\cup U,E)$ is said to be $(\eps,\gamma)$ small-set expanding if for any $S\subseteq V$ of size at most
  $\eps\card{V}$ we have that $\Phi(S)\geq \gamma$.
\end{definition}
Upon first sight, small-set expanders may seem like a slight variation over the combinatorial definition of expander graphs. However this seemingly small
change makes a big of difference, and indeed small-set expanders are more difficult to study. For instance, there is no general known equivalence between the
spectrum of the adjacency matrix of a graph and the small-set expansion of the graph. This is partly the reason that small-set expanders are typically much harder
to work with.

The main goal of this paper is to present the idea of ``low-dimension embeddings'' and show how to use it to give alternative, more direct proofs to several results
concerning computational problems over small set expanders. In particular, we
(1) recover the results of~\cite{TWZ,LRTV} about strong parallel repetition for Unique-Games over small set expanders,
(2) we show an improved Cheeger's inequality as well as approximation algorithm for Max-Cut over small set expanders, recovering results of~\cite{KLLOT,HigherOrder}.

Conceptually, we show that for several cases (such as the computational problems mentioned above), it is beneficial to consider an intermediate ``low-dimensional''
projection step which incurs less quantitative loss when compared to the usual integral rounding procedures. In a low dimension, one has additional tools in their
disposal to manipulate a vector valued solution so that it is possible to perform a final integral rounding procedure with better performance. We believe that such
ideas may be helpful towards getting improved analysis of other semi-definite programming relaxations of combinatorial optimization problems, and hope that the current
work encourages research along this direction.

\subsection{Our Results}
In this section, we formally state the main results of this paper.
\subsubsection{Solving Max-Cut on Small Set Expanders}
Our first result addresses the Max Cut problem. We show that on small set expanders, one can get an algorithm
achieving a better approximation ratio compared to the Geomans-Williamson algorithm~\cite{GW}.
\begin{thm}\label{thm:main_MC}
  There exists an absolute constant $C>0$, such that the following holds for all $\eps,\delta,\gamma>0$.
  There is an efficient algorithm for the following task:

  \noindent{\bf Input:} a regular, $(\eps,\gamma)$ small-set expander $G=(V\cup U,E)$ containing a cut of fractional size $1-\delta$.

  \noindent{\bf Output:} a bipartition $V = V_1\cup V_2$ of the vertices that cuts at least $1-C\frac{\delta}{\eps\gamma^6}$ fraction of the edges.
\end{thm}

\subsubsection{Spectral Partitioning}
Our second result is an improved quantitative for Cheeger's inequality when the underlying graph is a small set expander.
Recall that one side of Cheeger's inequality states that if $G = (V,E)$ is a regular graph such that $\lambda_2(G) \leq 1-\delta$, then one can
find in polynomial time a set $S\subseteq V$ of size at most $\card{V}/2$ such that $\Phi(S)\leq\sqrt{2\delta}$. The other side of Cheeger's inequality
asserts that if $\lambda_2(G) = 1-\delta$, then $\min_{\card{S}\leq \card{V}/2} \Phi(S)\geq \delta/2$, hence the upper bound is tight up to the square root
(and constant factors).
We show that for small set expanders,
one can considerably improve on the upper bound, and get a result which is much closer to the lower bound in Cheeger's inequality:
\begin{thm}\label{thm:main_cheeger}
  There exists an absolute constant $C>0$, such that if $G$ is a $(\eps,\gamma)$ small set expander
  and $\lambda_2(G)\geq 1-\delta$, then
  $G$ contains a set of size at most $n/2$ such that $\Phi(S)\leq C\cdot \frac{\delta}{\gamma^3\eps^3}$.
  Furthermore, such set $S$ can be found efficiently.
\end{thm}

\subsubsection{Parallel Repetition}
Our third result is about parallel repetition when applied on Unique-Games over small-set expanders. We begin by formally defining the Unique-Games problem.
\begin{definition}
  A Unique-Games instance is composed of a graph $G = (V,E)$, a finite alphabet $\Sigma$ and a collection of constraints $\{\phi_e\}_{e\in E}$, one for
  each edge. The constraint on an edge $e$ is a $1$-to-$1$ relation $\phi_e\subseteq \Sigma\times \Sigma$, specifying for each edge the tuples of
  labels that are considers satisfactory.
\end{definition}
Given a Unique-Games instance $\Psi = (G,\Sigma,\Phi = \{\phi_e\}_{e\in E})$, the value of $\Psi$, denoted by ${\sf val}(\Psi)$, is the maximum fraction
of constraints that may be satisfied by any assignment $A\colon V\to \Sigma$. A central conjecture in complexity theory, called the Unique-Games Conjecture~\cite{Khot},
asserts for any $\eps,\delta>0$, there is $k\in\mathbb{N}$ such that given a Unique-Games instance $\Psi$ with alphabet size $k$, it is NP-hard to distinguish
between the case that ${\sf val}(\Psi)\geq 1-\eps$ and the case that ${\sf val}(\Psi)\leq \delta$.
A well known approach to the Unique-Games Conjecture proceeds by first establishing a weak form of the conjecture, wherein $\delta$ is also close to $1$,
and then amplifying the gap via \emph{parallel repetition}. Here, given a Unique-Games instance $\Psi$ and a parameter $t\in\mathbb{N}$, the $t$-fold repeated game corresponds to the tensor product of the game is $\Psi^{\otimes t} = (G' = (V^t, E'), \Sigma^{t}, \Phi')$, where
\[
E' = \sett{((u_1,\ldots,u_t),(v_1,\ldots,v_t))}{(u_i,v_i)\in E~\forall i\in [t]},
\]
and $\Phi' = \{\phi_{e'}\}_{e'\in E'}$ where the constraint on edge $e'$ corresponds to the AND of the $t$ constraints on the $t$ edges of the original graph
$G$.

The parallel repetition theorem of Raz~\cite{Raz} and subsequent improvements~\cite{Holenstein,BravermanGarg,Rao}
turn out not to be good enough; such results are only able to show that for a unique game $\Psi$, if ${\sf val}(\Psi)\leq 1-\eps$, then the value of
the $t$-fold repeated game is at most $(1-\eps^2)^{\Omega(t)}$. This square  often times does not matter in application, however in the context
of Unique-Games it is crucial. Indeed, the weak forms of the Unique-Games conjecture that seem plausible may go to soundness which is as small as
$\delta = 1-\Theta(\sqrt{\eps})$ (though current technology does not even close to this), thereby offering a quadratic difference at best. This makes
the quadratic loss in parallel repetition unaffordable for the purposes of proving the Unique-Games Conjecture.

This raises the question: on which classes of graphs can one prove an
improved parallel repetition theorem, surpassing the square barrier?
Much effort had been devoted to study this question, starting with the tightness of the parallel repetition theorem~\cite{Raz,KROW,AK,FKO} and
subsequently studying parallel repetition over special classes of graphs~\cite{BHHRRS,AKKSTV,SS,BM21,TWZ,LRTV,MoshkovitzSSE}, most popularly on
expanders and small-set expanders. It is worth noting that while it is known that the Unique-Games Conjecture fails on such graph~\cite{AKKSTV,RagSteurer}, it is
still possible that there is a regime of parameters for which a weak form of the Unique Games Conjecture holds, and for which a strong version
of the parallel repetition theorem holds.

Using our technique, we are able to recover the result of~\cite{TWZ,LRTV}, and show that strong parallel repetition holds for unique games over small set expanders:
\begin{thm}\label{thm:main_UG}
  There exists an absolute constant $c>0$, such that the following holds for all $0 < \eps,\gamma,\delta<1$.
  If $\Psi = (G, \Sigma, \Phi)$ is
  a Unique-Games instance such that $G$ is $(\eps,\gamma)$ small-set expanding, regular graph whose
  value is at most $1-\delta$, then ${\sf val}(\Psi^{\otimes t})\leq (1-\delta)^{c\eps^3\gamma^3 t}$.
\end{thm}
We remark that our proof applies to the more general class of projection games, but we state it only for Unique-Games for the purpose of this introduction.

\subsection{Our Technique}
The proofs of Theorems~\ref{thm:main_MC} and~\ref{thm:main_cheeger} follow a similar theme.
In both cases, there are well known classical results that are very much similar to it except that they incur an additional square-root loss: Cheeger's inequality and
the Goemans Williamnson algorithm~\cite{GW}.
The source for this quadratic loss comes from a rounding phase in their proofs, which transforms a vector-valued solution with an $\ell_2^2$ distances
guarantee to an integral valued solution with an $\ell_1$ distances guarantee via rounding.
Our proofs of Theorems~\ref{thm:main_MC} and~\ref{thm:main_cheeger} avoid this quadratic loss by incorporating a preprocessing phase
and a ``soft'' rounding phase, which incur constant factors loses depending on the small set expansion parameters of the graphs, but ensures that the subsequent integral
rounding phase would be almost lossless.

\paragraph{Low-dimensional embeddings.} Our arguments also utilize the idea of low-dimensional embeddings. Here, one is given a collection of vectors (typically an SDP
solution to some problem, such as Max-Cut or Unique-Games) which a-priori may lie in a high dimensional Euclidean space. Common wisdom, which manifests itself in standard
rounding algorithms (such as Geomans-Williamson~\cite{GW}), utilize such vectors to generate $1$-dimensional Gaussians vectors with certain correlations. Our idea is to
use higher (but constant) dimensional projections as a way to reduce the dimension of the vector while roughly preserving the quality of the solution.
This allows us to apply ideas from low-dimensional geometry and construct a rounding scheme, which for small-set expanders, outperforms
the standard rounding schemes for these problems.

Similar, yet different ideas of low-dimensional projections have already made some appearances in the literature, in the context of higher order Cheeger inequalities~\cite{HigherOrder}
and in establishing improved bounds on the Grothendieck constant~\cite{BKMN}. 
Below, we explain in more detail our proof strategy for our results.

\subsubsection{Proof of Theorem~\ref{thm:main_MC}}
Our proof of Theorem~\ref{thm:main_MC} makes use of the Goemans-Williamson semi-definite program relaxation of Max-Cut~\cite{GW}. Our algorithm
proceeds by the following steps:
\begin{enumerate}
  \item {\bf Solving the SDP program:} We find unit vectors in $\mathbb{R}^m$ that achieve value at least $1-\delta$, where $m\leq n$. In the next
  two steps we pre-process these vectors to get a (different) collection of vectors that can be used in a rounding scheme with improved
  guarantees.

  \item {\bf Dimension reduction:} We show, via appropriate random projections, that we can reduce $m$ to $3$ as long as we are willing to take a multiplicative
  cost in the error. More precisely, we show we can find a solution consisting of unit vectors in $\mathbb{R}^3$ that achieve value at least $1-C\delta$, for some absolute value $C>0$.
  Denote them by $(x_v)_{v\in V}$.

  We remark that it is also possible to reduce the dimension down to $2$, and moreover that it simplifies subsequent steps.
  However, doing so may reduce the value of the resulting solution to $1-C\delta\log(1/\delta)$, and we would like to avoid this logarithmic factor.

  \item {\bf Soft rounding:} The main issue in the standard hyperplane rounding procedure in the Goemans-Williamson algorithm is the case that
  the vectors $x_v$ are uniformly distributed on the sphere, satisfying that $\norm{x_u-x_v}_2 = \Theta(\sqrt{\delta})$ for a typical edge $(u,v)$ in the graph.
  In this step, we show that in
  the case of small set expanders, this cannot be the case. To be more precise, we show that such case can be avoided if we are willing to pay constant
  factors depending on the expansion parameters of the graph.

  Towards this end, we show that given a low-dimension solution $(x_v)_{v\in V}$ from the previous step, one may produce a solution $(y_v)_{v\in V}$
  with similar value, such that the set of vertices $V$ may be divided into two sets, $Z_1,Z_2$, where vectors corresponding to vertices in $Z_2$
  occur in the union of at most $1/\eps$ balls with small $\ell_2$ radius, and for any edge $(u,v)$ that has
  at least one of its endpoints in $Z_1$ must either have $0$ or long $\ell_2$ distance, i.e. $\norm{y_v - y_u}_2 = 0$ or
  $\norm{y_v - y_u}_2\geq \eps$. Thus, in a sense the ``bad structure'' that may occur
  for the hyperplane rounding procedure in the Goemans-Williamson algorithm can only occur inside $Z_2$.

  \item {\bf Rounding:} Finally, we show an improved analysis of the standard hyperplane rounding when performed on a collection of vectors with
  the structure from the previous step. Intuitively, as $Z_2$ is small, we can show via a union bound that with constant probability,
  a random hyperplane doesn't cut any of the small balls in which vectors from $Z_2$ occur. Then, we sample a hyperplane conditioned on this event
  happening and show that a square loss does not occur for the rest of the edges (as the distance over these edges already has a
  ``Boolean-type'' behaviour of either being $0$, or far from $0$).
\end{enumerate}

\subsubsection{Proof of Theorem~\ref{thm:main_cheeger}}
To prove Theorem~\ref{thm:main_cheeger}, we view the second eigenvector of $G$, call it $x$ as an embedding of the graph into $\mathbb{R}$.
Recall that the square-root loss in the standard proof of Cheeger's inequality comes from an application of Cauchy-Schwarz, which
bounds $\sum\limits_{(u,v)\in E}{\card{x_v}\card{x_u - x_v}}$ by the square root of spectral gap of $G$. We wish to circumvent this loss,
and for that we observe that if the ratios of entries in $x$ were either $1$ or bounded away from $1$, i.e. outside the interval $[1-\eps,1+\eps]$,
then one may indeed proceed as
\[
\sum\limits_{(u,v)\in E}{\card{x_v}\card{x_u - x_v}}
=\sum\limits_{(u,v)\in E}{\card{x_v}^2\card{x_u/x_v - 1}}
\leq \sum\limits_{(u,v)\in E}{\card{x_v}^2\frac{\card{x_u/x_v - 1}^2}{\eps}}
=\frac{1}{\eps}\sum\limits_{(u,v)\in E}{\card{x_u - x_v}^2}.
\]

Thus, our goal is to preprocess $x$ so that we obtain this property, while at the same time not increasing
$\sum\limits_{(u,v)\in E}{\card{x_u - x_v}^2}$ by too much. Indeed, we show that a soft rounding strategy
inspired by the above algorithms can be used in this case. We partition $\mathbb{R}$ ``dyadically'' into
intervals of the form $[(1+\eps)^{i}, (1+\eps)^{i+1})$. We show that inside each one of these intervals, there is a way
to round each entry of $x$ to one of the endpoints without incurring too much of a loss in
$\sum\limits_{(u,v)\in E}{\card{x_u - x_v}^2}$.

\subsubsection{Proof of Theorem~\ref{thm:main_UG}}
The proof of Theorem~\ref{thm:main_UG} follows the information theoretic approach to parallel repetition~\cite{Raz,Rao,Holenstein,BravermanGarg}.
Our proof follows the outline of~\cite{BravermanGarg}, except that in the rounding phase in that result ---
which uses Pinkser's inequality to transform a KL-divergence closeness between distributions to a statistical distance closeness guarantee between the same distributions ---
we perform a preprocessing step. Indeed, we show that one may appeal to our improved version of Cheeger's inequality
to change the distributions so that the KL-divergence between them does not change too much and yet the square root loss in Pinsker's inequality does not occur.
From there,
the rest of our proof follows the outline of~\cite{BravermanGarg}.

\section{Solving Max-Cut on SSE's: Proof of Theorem~\ref{thm:main_MC}}
In this section, we prove Theorem~\ref{thm:main_MC} following the outline given in the introduction.
\subsection{The Goemans-Williamson Semi-definite Program Relaxation}
Below is the standard semi-definite program relaxation of Max-Cut. Instead of thinking of the graph $G$,
it will be more convenient for us to work with the graph $G' = (V',E')$, wherein each vertex has two copies,
$V' = V\times \{-1,1\}$, and the edges are $E' = \sett{((v,b),(u,b'))}{(v,u)\in E, b\neq b'}$. We have the following claim.
\begin{claim}
  The following two assertions hold:
  \begin{enumerate}
    \item If $G$ has a cut of fractional size $1-\delta$, then $G'$ has a cut of fractional size at least $1-\delta$.
    \item If $G$ is an $(\eps,\gamma)$ small set expander, then $G'$ is a $(\eps/2,\gamma/2)$ small set expander.
  \end{enumerate}
\end{claim}
\begin{proof}
  The first item is obvious. For the second item, suppose $S'\subseteq V'$ is a set of fractional size at most $\eps/2$, and
  let $S = \sett{v}{\exists b\in\{-1,1\},~ (v,b)\in S'}$. Then $S$ has fractional size at most $\eps$, and we may write
  \[
  \Phi(S')
  \geq \Prob{(v, b)\in S', (u,b')\text{ neighbour}}{(u,b')\not\in S'}
  \geq \Prob{(v, b)\in S', (u,b')\text{ neighbour}}{u\not\in S}.
  \]
  Note that sampling $(v,b)\in S'$, the distribution of $v$ may not be uniform over $S$, but the probability of each $v\in S$
  is either $1/\card{S'}$ or $2/\card{S'}$, hence
  \[
  \Prob{(v, b)\in S', (u,b')\text{ neighbour}}{u\not\in S}
  \geq
  \frac{1}{2}\Prob{v\in S, u\text{ neighbour}}{u\not\in S}
  =\frac{1}{2}\Phi(S)
  \geq \frac{\gamma}{2}.\qedhere
  \]
\end{proof}

We thus write the program below, which is the standard semi-definite program relaxation for Max-Cut for the graph $G'$:
\[
\begin{array}{lll}\label{prog:GW}
  \min & \frac{1}{\card{E'}}\sum\limits_{((u,b),(v,b'))\in E'}\norm{x_{u,b} - x_{v,b'}}_2^2 & ~\\
  \text{subject to} & \norm{x_{v,b}}_2 = 1 &~\forall (v,b)\in V',\\
  &x_{(v,b)}\in \mathbb{R}^{2n} &~\forall (v,b)\in V',\\
  &x_{(v,-b)}  = -x_{(v,b)} &~\forall (v,b)\in V'.
\end{array}
\]
It is clear that if $G'$ contains a cut of fractional size at least $1-\delta$, then there a solution for
the above program with value at most $\delta$. Using the ellipsoid algorithm, we may efficiently find
a collection of vectors $(x_v)_{v\in V}$ achieving value at most $\delta + \xi = \delta'$ where $\xi$
decays with the runtime of the algorithm; we take $\xi = \delta$ and fix such solution henceforth

\subsection{Reducing the Dimension of the Semi-Definite Program Solution}\label{sec:dim_reduce_3}
The goal of this section is to prove the following lemma.
\begin{lemma}\label{lem:dim_reduce}
  There exists an absolute constant $C>0$ such that the following holds. There is an efficient, randomized procedure
  that given a solution $(x_v)_{v\in V'}$ to SDP program to Max-Cut with value at most $\delta'$ consisting of vectors
  from $\mathbb{R}^{n}$, outputs a solution $(z_v)_{v\in V'}$ consisting of vectors from $\mathbb{R}^3$, whose value is
  at most $C\delta'$.
\end{lemma}
\begin{proof}
  Let $g_1,g_2,g_3\sim {\sf N}(0,I_n)$ be independent multi-dimensional Gaussians. For each $v\in V$, define
  \[
  y_v = \left(\inner{x_v}{g_1}, \inner{x_v}{g_2}, \inner{x_v}{g_3}\right),
  \qquad\qquad
  z_v = \frac{y_v}{\norm{y_v}_2}.
  \]
  Clearly, $(z_v)_{v\in V'}$ is a solution to the SDP program, and we next analyze the objective value that it gets.
  Fix $u,v\in V'$, and let $\delta_{u,v} = \norm{x_u - x_v}_2^2$. We will show that, $\Expect{g_1,g_2,g_3}{\norm{z_v - z_u}_2^2} = O(\delta_{u,v})$
  from which the claim follows by linearity of expectation over all the edges of the graph.

  Denote $z = \frac{y_v}{\norm{y_u}_2}$, then
  \begin{align*}
  \Expect{g_1,g_2,g_3}{\norm{z_v - z_u}_2^2}
  &\leq
  2\Expect{g_1,g_2,g_3}{\norm{z_v - z_u}_2^2 1_{\norm{y_v}_2\leq \norm{y_u}_2}}\\
  &\leq
  4\underbrace{\Expect{g_1,g_2,g_3}{\norm{z_v - z}_2^2 1_{\norm{y_v}_2\leq \norm{y_u}_2}}}_{(\rom{1})} +
  4\underbrace{\Expect{g_1,g_2,g_3}{\norm{z_u - z}_2^2 1_{\norm{y_v}_2\leq \norm{y_u}_2}}}_{(\rom{2})},
  \end{align*}
  and we upper bound each expectation separately, each by $O(\delta_{u,v})$, thereby finishing the proof.
  \begin{claim}
    $(\rom{2})\leq O(\delta_{u,v})$.
  \end{claim}
  \begin{proof}
  We write
  \begin{align*}
  \Expect{g_1,g_2,g_3}{\norm{z_u - z}_2^2 1_{\norm{y_v}_2\leq \norm{y_u}_2}}
  &=\Expect{g_1,g_2,g_3}{\frac{\norm{y_u - y_v}_2^2}{\norm{y_u}_2^2} 1_{\norm{y_v}_2\leq \norm{y_u}_2}}\\
  &=
  \Expect{g_1,g_2,g_3}{\frac{\norm{y_u - y_v}_2^2}{\norm{y_u}_2^2} 1_{\norm{y_v}_2\leq \norm{y_u}_2} 1_{\norm{y_u}_2< \delta_{u,v}}}\\
  &+
  \Expect{g_1,g_2,g_3}{\frac{\norm{y_u - y_v}_2^2}{\norm{y_u}_2^2} 1_{\norm{y_v}_2\leq \norm{y_u}_2} 1_{\norm{y_u}_2\geq \delta_{u,v}}}.
  \end{align*}
  For the first expectation, as $\norm{y_v}_2\leq \norm{y_u}_2$ we have
  $\frac{\norm{y_u - y_v}_2^2}{\norm{y_u}_2^2}\leq 4$, so the first expectation is bounded by
  \[
  4\Expect{g_1,g_2,g_3}{1_{\norm{y_v}_2\leq \norm{y_u}_2} 1_{\norm{y_u}_2< \delta_{u,v}}}
  \leq
  4\Expect{g_1,g_2,g_3}{1_{\norm{y_u}_2< \delta_{u,v}}}
  \leq 4 \Prob{G\sim N(0,1)}{\card{G}\leq \delta_{u,v}}
  =O(\delta_{u,v}).
  \]
  For the second expectation, we upper bound it by
  \[
  \sum\limits_{k=0}^{\infty}
  \Expect{g_1,g_2,g_3}{\frac{\norm{y_u - y_v}_2^2}{2^{2k} \delta_{u,v}^2} 1_{2^{k}\delta_{u,v}\leq \norm{y_u}_2< 2^{k+1}\delta_{u,v}}}.
  \]
  We bound the tail, that is the sum over $k\geq \log(1/\delta_{u,v})$, by
  \[
  \sum\limits_{k=\lceil{\log(1/\delta_{u,v})\rceil}}^{\infty}
  \Expect{g_1,g_2,g_3}{\norm{y_u - y_v}_2^2 1_{2^{k}\delta_{u,v}\leq \norm{y_u}_2< 2^{k+1}\delta_{u,v}}}
  \leq\Expect{g_1,g_2,g_3}{\norm{y_u - y_v}_2^2}
  =O(\norm{x_u-x_v}_2^2) = O(\delta_{u,v}).
  \]

  As for the sum on $k=0,\ldots,\lceil{\log(1/\delta_{u,v})\rceil}-1$, take $p\in\mathbb{N}$ sufficiently large ($p=4$ will do), and let $p'$ be the H\"{o}lder conjugate of $p$.
  Then
  \begin{align}
  &\Expect{g_1,g_2,g_3}{\frac{\norm{y_u - y_v}_2^2}{2^{2k} \delta_{u,v}^2} 1_{2^{k}\delta_{u,v}\leq \norm{y_u}_2< 2^{k+1}\delta_{u,v}}}\notag\\\label{eq1}
  &\qquad\leq
  \frac{1}{2^{2k}\delta_{u,v}^2}
  \left(\Expect{g_1,g_2,g_3}{\norm{y_u - y_v}_2^{2p}}\right)^{1/p}
  \left(\Expect{g_1,g_2,g_3}{1_{2^{k}\delta_{u,v}\leq \norm{y_u}_2< 2^{k+1}\delta_{u,v}}}\right)^{1/p'}.
  \end{align}
  First, we have
  \[
  \Expect{g_1,g_2,g_3}{\norm{y_u - y_v}_2^{2p}}
  =\Expect{g_1,g_2,g_3}{\left(G_1^2+G_2^2+G_3^2\right)^p}
  \leq 3^{p-1}\Expect{g_1,g_2,g_3}{G_1^{2p}+G_2^{2p}+G_3^{2p}}
  =3^p \Expect{g_1,g_2,g_3}{G_1^{2p}},
  \]
  where $G_1 = \inner{g_1}{x_u-x_v}$, $G_2 = \inner{g_2}{x_u-x_v}$, $G_3 = \inner{g_3}{x_u-x_v}$.
  As $G_1$ is a Gaussian random variable with mean $0$ and variance $\delta_{u,v}$, we know that
  \[
  \Expect{g_1,g_2,g_3}{G_1^{2p}} = O_p(\delta_{u,v}^p).
  \]

  Second,
  \[
  \Expect{g_1,g_2,g_3}{1_{2^{k}\delta_{u,v}\leq \norm{y_u}_2< 2^{k+1}\delta_{u,v}}}
  \leq \Prob{g_1,g_2,g_3}{\card{\inner{g_1}{x_u}},\card{\inner{g_2}{x_u}},\card{\inner{g_3}{x_u}}\leq 2^{k+1}\delta_{u,v}},
  \]
  and as $g_1,g_2,g_3$ are independent, the last probability is at most $2^{3(k+1)}\delta_{u,v}^3$.

  Plugging everything into~\eqref{eq1} we get that the sum on $k=0,\ldots,\lceil{\log(1/\delta_{u,v})\rceil}-1$ is at most
  \[
  \sum\limits_{k=0}^{\lceil{\log(1/\delta_{u,v})\rceil}-1}
  \frac{1}{2^{2k}\delta_{u,v}^2}
  O_p(\delta_{u,v})
  \left(2^{3(k+1)}\delta_{u,v}^3\right)^{1/p'}
  =\delta^{3/p'-1}O_p(1)\sum\limits_{k=0}^{\lceil{\log(1/\delta_{u,v})\rceil}-1}
  2^{(3/p'-2)k}.
  \]
  For sufficiently large $p$, $3/p' - 2 > 0$, and we may bound the last sum by
  $O(2^{(3/p'-2)\lceil{\log(1/\delta_{u,v})}\rceil})=\delta_{u,v}^{2-3/p'}$, and plugging
  that in yields the bound $O_p(\delta_{u,v})$.
  \end{proof}

  \begin{claim}
    $(\rom{1})\leq O(\delta_{u,v})$.
  \end{claim}
  \begin{proof}
  We write
  \begin{align*}
  \Expect{g_1,g_2,g_3}{\norm{z_v - z}_2^2 1_{\norm{y_v}_2\leq \norm{y_u}_2}}
  &=\Expect{g_1,g_2,g_3}{\norm{y_v\left(\frac{1}{\norm{y_v}_2} - \frac{1}{\norm{y_u}_2}\right)}_2^2 1_{\norm{y_v}_2\leq \norm{y_u}_2}}\\
  &=
  \Expect{g_1,g_2,g_3}{\left(1 - \frac{\norm{y_v}_2}{\norm{y_u}_2}\right)^2 1_{\norm{y_v}_2\leq \norm{y_u}_2}}\\
  &=\Expect{g_1,g_2,g_3}{\left(\frac{\norm{y_u}_2-\norm{y_v}_2}{\norm{y_u}_2}\right)^2 1_{\norm{y_v}_2\leq \norm{y_u}_2}}\\
  &\leq
  \Expect{g_1,g_2,g_3}{\left(\frac{\norm{y_u - y_v}_2}{\norm{y_u}_2}\right)^2 1_{\norm{y_v}_2\leq \norm{y_u}_2}}\\
  &=(\rom{1}),
  \end{align*}
  hence the claim follows from the prior claim. We used the fact that $\card{\norm{a}-\norm{b}}\leq \norm{a-b}$.
\end{proof}
\end{proof}

\subsection{The Soft Rounding: the Cubicular Rounding}\label{sec:MC_cubic_round}
In this section, we pick off where the last section ended with a solution vector collection $(z_v)_{v\in V'}$ consisting
of vectors from $\mathbb{R}^3$ that has value $\delta''$ where $\delta''\leq O(\delta') = O(\delta)$. Our goal is to construct a new vector solution
$\{z_v''\}_{v\in V'}$ whose value is at most $\delta'''$ for $\delta''' = O(\delta''/\gamma)$ that satisfies the following property.
The set of vectors $\{z_v''\}_{v\in V'}$ can be partitioned into two collections, $Z_1$ and $Z_2$, such that
$\card{Z_2}\leq 1/\eps$, and for any edge $(u,v)$ such that $z_u''\neq z_v''$ and $z_u''\in Z_1$ it holds that $\norm{z_u'' - z_v''}_2\geq \Omega(\eps)$.
In Section~\ref{sec:round} we show how to use this property to perform a more efficient
analysis of the standard hyperplane rounding.

\subsubsection{Partitioning the $3$-dimensional sphere into cubes}
Consider the $3$ dimensional sphere $S^2 = \sett{z\in\mathbb{R}^3}{\norm{z}_2 = 1}$. We will use a simple construction of triangulation
of the sphere (geodesic polyhedron), that we describe next. Starting with a standard cube in $\mathbb{R}^2$,
denoted it by $\mathcal{T}_6$ (that has $6$ faces and $8$ vertices), we may partition each face of it (which is a square) into $4$
squares in the natural way (i.e. by adding the lines that connect the midpoints of opposing sides).
We then project the $4$ new points that were generated by each face into $S^2$,
and get a more refined regular polygon with square faces, with $4$ times as many faces; denote it by $\mathcal{T}_{24}$. We will perform this operation
several times to reach a sufficiently refined shape; we need the perimeter of each face to be small compared to $\eps$, say at most $\frac{\eps}{K}$
for sufficiently large absolute constant $K$. We note that at each iteration, the perimeter of the faces shrinks by factor $2$, hence
we may pick a power of $4$, $t\geq \frac{K'}{\eps^2}$ for $K'$ depending only on $K$, and have that each phase of $\mathcal{T}_t$ has
perimeter at most $\frac{\eps}{K}$. We summarize this discussion with the following standard fact.

\begin{fact}\label{fact:triangulation}
  For all $K>0$ there is $K' = K'(K)$, such that for all $\eps>0$,
  taking $t$ the smallest number of the form $6\cdot 4^{m}$ that is larger than $K'/\eps^2$,
  there is a regular polygon $\mathcal{T}_t\subseteq \mathbb{R}^3$ that can be constructed in $O_{\eps}(1)$ time
  such that:
  \begin{enumerate}
    \item Each face of $\mathcal{T}_t$ is a square.
    \item Central symmetry: $P$ is a face of $\mathcal{T}_t$ if and only if $-P$ is a face of $\mathcal{T}_t$.
    \item $\mathcal{T}_t$ has $t$ faces and $\frac{4}{3} t$ vertices.
    \item The perimeter of each face is at most $\eps/K$.
  \end{enumerate}
\end{fact}

\subsubsection{Division of the sphere}
Using the triangulation $\mathcal{T}_t$ from Fact~\ref{fact:triangulation}, we may define a division of $S^2$
by projecting each face of $\mathcal{T}_t$ into the sphere. This way, we get a partition $P_1\cup\ldots\cup P_t$
of $S^2$. We also note that the perimeter of each of $P_i$ is the same, and is at most order of the perimeter
of a face from $\mathcal{T}_t$, i.e. at most $O(\eps/K)$.

For each $i$ define the neighbourhood of $P_i$ by:
\[
P_i^{\uparrow} = \sett{u\in S^{2}}{\exists p\in P_i, \norm{u-p}_2\leq \eps}.
\]
For each subset of $A\subseteq S^2$, we define the mass of $A$ as ${\sf mass}(A) = \card{\sett{v\in V}{z_v\in A}}$.

\begin{claim}\label{claim:mass_sum}
  $\sum\limits_{i=1}^t {\sf mass}(P_i^{\uparrow}) = O(n)$.
\end{claim}
\begin{proof}
  Consider the bipartite graph $\tilde{G} = ([t]\times [t], \tilde{E})$ wherein we connect
  $i$ to $j$ if $P_i^{\uparrow} \cap P_j\neq\emptyset$. We note that the degree of
  each vertex in $\tilde{G}$ is $O(1)$, and $P_i^{\uparrow} \subseteq \cup_{j:(i,j)\in\tilde{E}} P_j$.
  Thus,
  \[
  \sum\limits_{i=1}^t {\sf mass}(P_i^{\uparrow})
  \leq
  \sum\limits_{i=1}^t \sum\limits_{j:(i,j)\in \tilde{E}}{\sf mass}(P_j)
  =\sum\limits_{j=1}^{t} {\sf deg}(j){\sf mass}(P_j)
  =O\left(\sum\limits_{j=1}^{t} {\sf mass}(P_j)\right)
  =O(n).\qedhere
  \]
\end{proof}
It follows that for all but at most $2/\eps$ of the $i$'s we have that ${\sf mass}(P_i\uparrow)\leq \eps n/2$.
We refer to such $i$'s as light parts of the partition. In the following section we show that one may change
a vector $z_v$ such that $v$ is in $P_i$ for a light $i$, to be a vector in the boundary of $P_i$, and only incur a
constant factor loss in the objective value, effectively reducing these vectors to be $2$-dimensional.

\subsubsection{Rounding to the Boundary}
We now modify the solution $(z_u)_{u\in V}$ to $(z'_u)_{u\in V}$ as follows. For $u$, consider the $i$ such that $z_u\in P_i$. If $P_i$ is heavy
we set $z'_u = z_u$, and otherwise we modify it as follows. Let $p_i$ be the center of $P_i$, and consider
the segment from $p_i$ to $z_u$, and in particular its two intersection points with the boundary of $P_i$. We
let $z'_u$ be the intersection point closer to $z_u$ among these $2$.
\begin{lemma}\label{lem:round_to_boundary}
  The collection of vectors $z_u'$ defined about is a solution to the Max-Cut program, and
  \[
  \frac{1}{\card{E'}}\sum\limits_{(u,v)\in E'}\norm{z'_{u} - z'_{v}}_2^2\leq \delta'''
  \]
  for $\delta''' = O(\delta''/\gamma^3)$.
\end{lemma}
\begin{proof}
  First, it is clear that all of the $z_u$'s are norm $1$ vectors, and also that they respect the conditions
  $z_{u,-b} = - z_{u,b}$ as the part of a vector $z$ and its negation have the same mass (as if $P_i$ is the
  part of $z$, then $-P_i$ is the part of $z$).

  Fix a light $P_i$ and denote ${\sf cost}_i = \sum\limits_{\substack{u: z_u\in P_i^{\uparrow} \\ (u,v)\in E'}}{\norm{z_u - z_v}_2^2}$.
  We argue that ${\sf cost}_i$ grows multiplicatively by at most factor $O(1/\gamma^3)$ due to the change of the procedure on the vectors
  inside $P_i$. This finishes the proof,
  as the only other effect of the procedure we need to make note of is that for each edge $e=(u,v)\in E'$ such that $u\in P_i$, $v\in P_j$,
  $\norm{z'_{u} - z'_{v}}_2 = O(\norm{z_{u} - z_{v}}_2)$, so each edge in ${\sf cost}_j$ increases by a constant factor at most once by
  the procedure on other $P_i$'s.

  We now argue that ${\sf cost}_i$ grows multiplicatively by at most factor $O(1/\gamma^3)$, and assume towards contradiction
  this is not the case. For simplicity of notation, we drop the subscript $i$ as we will only be concerned with $P_i$ from now on.
  Let $B_0 = \sett{u}{z_u\in P}$, and let $W>0$ be a large absolute constant to be determined.
  As a result of our soft rounding, the contribution of edges whose endpoints are both outside $B_0$ grows by at most factor $O(1)$,
  and the contribution of edges that have at least one of their endpoints in $B_0$ grows additively by at most $2d\card{B_0}\eps^2$.
  Hence, if ${\sf cost} \geq  \frac{\gamma^3\eps^2}{50 W^2} d\card{B_0}$ we would have be done, so we assume from now on that
  ${\sf cost} <  \frac{\gamma^3\eps^2}{50 W^2} d\card{B_0}$

  Consider $\ell_1 = \frac{\gamma}{W}\eps$ and let $B_1 = \sett{u}{\exists p\in B_0, \norm{z_u-p}_2\leq \ell_1}$.
  Also, define
  $\ell_2 = 2\ell_1$ and set $B_2 = \sett{u}{\exists p\in B_0, \norm{z_u-p}_2\leq \ell_2}$.
  As $B_1\subseteq P^{\uparrow}$, ${\sf mass}(B_1)\leq \eps n$ so by
  small set expansion at least $\gamma d\card{B_1}$ of the edges touching $B_1$ escape
  it. If $\card{B_2\setminus B_1}\leq \frac{\gamma}{2} \card{B_1}$, we would get that
  at least $\frac{\gamma}{2} d\card{B_1}$ of the edges touching $B_1$ go outside $B_2$,
  hence
  \[
  {\sf cost}\geq \frac{\gamma}{2} d\card{B_1}(\ell_2-\ell_1)^2 = \frac{\gamma \ell_1^2}{2} d\card{B_1} > \frac{\gamma^3\eps^2}{50 W^2} d\card{B_0}
  \]
  and contradiction. Hence $\card{B_2}\geq (1+\gamma/2)\card{B_1}$. Next, define
  $\ell_3 = \ell_2 + \sqrt{\frac{\card{B_0}}{\card{B_2}}}\ell_1$ as well as
  $B_3 = \sett{u}{\exists p\in B_0, \norm{z_u-p}\leq  \ell_3}$. Then given $\ell_2\leq \eps$,
  we have that ${\sf mass}(B_2)\leq \eps$ and similar argument to before gives that
  $\card{B_3}\geq (1+\gamma/2)\card{B_2}$.
  Indeed, otherwise we get that at least $\frac{\gamma}{2}d\card{B_2}$ of the edges touching $B_2$
  go outside $B_3$, and then
  \[
  {\sf cost} \geq \frac{\gamma}{2}d\card{B_2}(\ell_3-\ell_2)^2 = \frac{\gamma\ell_1^2}{2} d\card{B_0},
  \]
  and contradiction. This way, we continue iteratively, and once $\ell_k$ and $B_k$ have been defined we take
  \[
  \ell_{k+1} = \ell_k + \sqrt{\frac{\card{B_{0}}}{\card{B_{k}}}}\ell_1,
  \qquad
  B_{k+1} = \sett{u}{\exists p\in B_0, \norm{z_u-p}\leq  \ell_{k+1}}
  \]
  and provided that $\ell_k\leq \eps$ we conclude that $\card{B_{k+1}}\geq (1+\gamma/2)\card{B_k}$. Note that
  provided we have gotten to the $k+1$ step, we have
  \[
  \ell_k \leq \sum\limits_{r=0}^{\infty} (1+\gamma/2)^{-r/2} \ell_1\leq O(\ell_1/\gamma)=O(\eps/W)\leq \eps
  \]
  for sufficiently large absolute constant $W$, so we can ensure that the argument goes through indefinitely.
  This is a contradiction, as at some point the size of $B_k$ would exceed $\eps n$.
\end{proof}

\subsubsection{Rounding to the Corners}
Next, we modify the vector-valued solution $(z'_u)_{u\in V}$ to a different vector-valued solution $(z''_u)_{u\in V}$ wherein (roughly speaking) vectors belonging to vertices on light $P_i$'s can only be in the corners of $\mathcal{T}_t$.

Towards this end, consider the boundary of the partition $\mathcal{P}$ and a particular arc in it.
This arc has $2$ neighbouring cells, say $P_i$ and $P_j$. If $P_i$ and $P_j$ are light, we look at that arc, mark its middle point $p_{i,j}$. We then round the
each vector $z'_u$ on that arc to the closer of the two endpoints of the arc to it (namely, according to the side of $z'_u$ with respect to $p_{i,j}$).

\begin{lemma}\label{lem:round_to_corners}
  The collection of vectors $z_u''$ defined about is a solution to the Max-Cut program, and
  \[
  \frac{1}{\card{E'}}\sum\limits_{(u,v)\in E'}\norm{z''_{u} - z''_{v}}_2^2\leq \delta'''
  \]
  for $\delta'''' = O(\delta'''/\gamma^3)$.
\end{lemma}
\begin{proof}
  Consider a single arc on which we perform the operation, denote it by $\mathcal{L}$, and let
  \[
  {\sf cost}_{\mathcal{L}} = \sum\limits_{\substack{(u,v)\in E: z'_u\in \mathcal{L}}}{\norm{z'_u - z'_v}_2^2}.
  \]
  We argue that as a result of the above operation, ${\sf cost}_{\mathcal{L}}$ increases by factor at most $O(1/\gamma^{3})$. This quickly finishes
  the proof since the only other effect of the operation we need to take note of, is the effect of other arcs on ${\sf cost}_{\mathcal{L}}$.
  For that, it suffices to observe that for each edge $(u,v)\in E$ such that $u\in\mathcal{L}$ and $v\in\mathcal{L}'$, the quantity
  $\norm{z'_u - z'_v}_2^2$ increases by factor at most $O(1)$ due to the operation on $\mathcal{L}'$.

  Let $L$ be the length of the arc $\mathcal{L}$ (noting it is of the order of $\eps$), $p$ be its midpoint and let $W>0$ be a large absolute
  constant. Let $\ell_0 = \frac{\gamma L}{W}$, and consider
  $B_0 = \sett{u}{z'_u\in\mathcal{L}, \norm{z'_u - p}\leq \ell_0}$. As a result of our operation, ${\sf cost}_{\mathcal{L}}$ changes as follows: the contribution of
  edges that have both endpoints outside $B_0$ grows by factor $O(1)$ at most, whereas the contribution of edges that have an endpoint inside $B_0$
  increases additively by at most $L^2 d\card{B_0}$. Hence, if ${\sf cost}_{\mathcal{L}}\geq \frac{\gamma^3 L^2}{50 W^2} d\card{B_0}$ we are done,
  so assume otherwise.

  Let $\ell_1 = 2\ell_0$, and note that $m(B_0)\leq m(\mathcal{L})\leq \eps n$, so by small set expansion at least $\gamma d\card{B_0}$ of the edges
  touching $B_0$ go outside it. At most $d\card{B_1\setminus B_0}$ of these edges go to $B_1$, hence if $\card{B_1\setminus B_0}\leq \gamma/2 \card{B_0}$
  we get that
  \[
  {\sf cost}_{\mathcal{L}}\geq \frac{1}{2}\gamma d\card{B_0}(\ell_1 - \ell_0)^2 > \frac{\gamma^3 L^2}{50 W^2} d\card{B_0},
  \]
  and contradiction. Thus, $\card{B_1\setminus B_0}> \gamma/2 \card{B_0}$, and so $\card{B_1}\geq (1+\gamma/2)\card{B_0}$.
  Continuing this way, iteratively define
  \[
  \ell_{j+1} = \ell_j + \sqrt{\frac{\card{B_{0}}}{\card{B_{j}}}}\ell_1,
  \qquad
  B_{j+1} = \sett{u}{z'_u\in\mathcal{L}, \norm{z'_u - p}\leq \ell_{j+1}}.
  \]
  If $\ell_{j+1}\leq L$ we can apply the small set expansion argument again to argue that $\card{B_{j+1}}\geq (1+\gamma/2)\card{B_j}$.
  Thus, provided that $\ell_{j}\leq L$ we may continue with the argument indefinitely and keep on increasing $B_j$, which is clearly impossible, giving
  us the desired contradiction. We note that if $W$ is large enough, then indeed we always have $\ell_j\leq L$, as we may bound
  \[
  \ell_{j}\leq \sum\limits_{r=0}^{j}\sqrt{\frac{\card{B_{0}}}{\card{B_{r}}}}\ell_1\leq \sum\limits_{r=0}^{j} (1+\gamma/2)^{-r} \ell_1 = O(\ell_1/\gamma) = O(L/W) < L
  \]
  for large enough $W$.
\end{proof}

\subsection{The Integral Rounding Procedure}\label{sec:round}
We next describe the integral rounding procedure using the vector solution $z''_u$ from Lemma~\ref{lem:round_to_corners}.
Let ${\sf Heavy}$ be the collection of all parts $P_i$ that are not light, denote by ${\sf Heavy}'$ be the collection of
all $P_j$'s adjacent to a face in ${\sf Heavy}$, and let $\mathcal{H}$ be the union of boundaries of faces in ${\sf Heavy}\cup{\sf Heavy}'$.

Consider choice of a hyperplane $H\subseteq \mathbb{R}^3$ uniformly at random.

\begin{claim}\label{claim:prob_to_intersect_skel}
  There exists $K>0$, such that in the above set-up,
  \[
  \Prob{H}{\exists \mathcal{L}\in \mathcal{H} \text{ such that }H\cap \mathcal{L}\neq \emptyset} \leq 1/2.
  \]
\end{claim}
\begin{proof}
By Claim~\ref{claim:mass_sum}, the number of heavy $P_i$'s is $O(1/\eps)$, and as each one of them has $O(1)$ many
arcs in its boundary, we deduce that $\card{{\sf Heavy}\cup {\sf Heavy}'} = O(1/\eps)$. As each face has
$O(1)$ arcs in the boundary, we conclude that $\card{\mathcal{H}} = O(1/\eps)$. The length of each $\mathcal{L}\in\mathcal{H}$ is
at most $\eps/K$, so we get that the probability that $H$ intersects a given arc is at most $O(\eps/K)$. Thus, by
the Union bound the probability that $H$ intersects some arc in $\mathcal{H}$ is at most $O(1/K)\leq 1/2$, provided that $K$ is large enough.
\end{proof}

We now sample $H$ conditioned on it not intersecting any face or arc from ${\sf Heavy}$. We claim that, for any edge $(u,v)$
such that $H$ intersects the line between $z''_u$ and $z''_v$, it must be the case that $\norm{z''_u-z''_v}_2\geq \Omega(\eps)$.
There are a few cases to consider:
\begin{enumerate}
  \item Either $z''_u$ or $z''_v$ is not a corner of $\mathcal{P}$, say $z''_u$. In this case, $z_u$ is in a face $P_j$ that is either
  heavy or adjacent to a heavy $P_i$. In either case, to cross the line between $z''_u$ and $z''_v$ the hyperplane $H$ must cross the
  boundary of $P_j$, but this does occur due to the conditioning on $H$.

  \item Else, $z''_u$ and $z''_v$ are two corners of $\mathcal{P}$, so either they are the same corner, or at least $\Omega(\eps)$ apart in $\ell_2$-distance.
\end{enumerate}

We consider the cut $S = \sett{v\in V}{\inner{z''_{v,1}}{x_H}\geq 0}$, and prove in the following lemma that the expected size of the
cut is at least $\left(1-O\left(\frac{\delta''''}{\eps}\right)\right)\card{E}$. Given such guarantee,
standard techniques show that such cut may be found in polynomial time by a randomized algorithm, thereby finishing the proof of Theorem~\ref{thm:main_MC}.
\begin{lemma}\label{lem:compute_rounding}
  $\cExpect{H}{H\text{ does not hit $\mathcal{H}$}}{{\sf Cut}(S)}\geq \left(1-O\left(\frac{\delta''''}{\eps}\right)\right)\card{E}$.
\end{lemma}
\begin{proof}
  The expected size of the cut is
  \[
  \sum\limits_{v,u} \cExpect{H}{H\text{ does not hit $\mathcal{H}$}}{1_{H\text{ separates } z_{u,1}''\text{ and }z_{v,1}''}},
  \]
  so the deficit in the cut size (i.e. $\card{E}$ minus its size) has expectation
  \[
  \sum\limits_{v,u} \cExpect{H}{H\text{ does not hit $\mathcal{H}$}}{1_{H\text{ separates } z_{u,1}''\text{ and }z_{v,-1}''}}.
  \]
  For $H$ to separate $z_{u,1}''$ and $z_{v,-1}''$, by the discussion proceeding the lemma, we must have
  $\norm{z_{u,1}'' - z_{v,-1}''}_2\geq \Omega(\eps)$, so we get that the last sum is at most
  \[
  2\sum\limits_{v,u} \Expect{H}{1_{H\text{ separates } z_{u,1}''\text{ and }z_{v,-1}''}} 1_{\norm{z_{u,1}'' - z_{v,-1}''}_2\geq \Omega(\eps)}.
  \]
  Here, we used the fact that the probability that $H$ does not hit $\mathcal{H}$ is at least $1/2$. Letting
  $\theta(z_{u,1},z_{v,-1})$ be the angle between the two vectors, we have
  \[
  \Expect{H}{1_{H\text{ separates } z_{u,1}''\text{ and }z_{v,-1}''}}
  =\frac{1}{\pi} \theta(z_{u,1}'',z_{v,-1}'')
  = O(\norm{z_{u,1}'' - z_{v,-1}''}_2),
  \]
  hence the expected deficit in the size of the cut defined by $S$ is at most
  \[
  O\left(\sum\limits_{v,u} \norm{z_{u,1}'' - z_{v,-1}''}_2 1_{\norm{z_{u,1} - z_{v,-1}}_2\geq \Omega(\eps)}\right)
  =\frac{1}{\eps} O\left(\sum\limits_{v,u} \norm{z_{u,1}'' - z_{v,-1}''}_2^2\right)
  =O\left(\frac{\delta''''}{\eps}\card{E}\right).
  \]
\end{proof}

\section{Spectral Partitioning}
In this section, we use our techniques to establish stronger algorithmic spectral partitioning results for the second largest eigenvalue of a graph $G$, proving
Theorem~\ref{thm:main_cheeger}. We also state an analogous result for the smallest eigenvalue of a small set expander $G$ (but omit the mostly-identical proof).

\subsection{Cheeger's Inequality on SSE's}
As explained in the introduction, to avoid the square root loss in Cheeger's inequality, we must make sure that we perform the rounding phase of it on
vectors in which the ratio between any two coordinates is either $1$, or is bounded away from $1$. Below, we formalize this notion via a notion we refer
to as $\eps$-quantization.
\begin{definition}
  Let $\eps>0$. We say a vector $x\in\mathbb{R}^n$ is $\eps$-quantized if
  for each $i,j\in [n]$:
  \begin{enumerate}
    \item $x(i) = 0$ or $x(j) = 0$;
    \item otherwise, $\card{x(i)} = \card{x(j)}$, $\card{x(i)}\geq (1+\eps)\card{x(j)}$, or $\card{x(i)}\leq (1-\eps)\card{x(j)}$.
  \end{enumerate}
\end{definition}

Next, we prove that given a SSE graph $G$ and a vector $x$ such that $\Expect{(u,v)\sim E}{(x(u) - x(v))^2}$ is small, one can transform $x$ to a vector
$y$ that is quantized while not increasing the square distances by much.
\begin{lemma}\label{lem:cheeger_quantized}
  Suppose that $G = (V,E)$ is a $d$-regular, $(\eps,\gamma)$ small set expander.
  If $x\in \mathbb{R}^n$ is a vector such that $\Expect{(u,v)\sim E}{(x(u) - x(v))^2}\leq \delta$, then
  there is an $\eps$-quantized vector $y\in \mathbb{R}^n$ such that:
  \begin{enumerate}
    \item $\Expect{(u,v)\sim E}{(y(u) - y(v))^2}\leq O\left(\frac{\delta}{\gamma^3\eps^2}\right)$.
    \item For each $v\in V$, $y(v) = 0$ if and only if $x(v) = 0$.
    \item For each $v\in V$, $(1-\eps)x(v)\leq y(v)\leq (1+\eps) \leq x(v)$ if $x(v)\geq 0$, and $(1+\eps)x(v)\leq y(v)\leq (1-\eps)x(v)$ otherwise.
  \end{enumerate}
\end{lemma}
\begin{proof}
  Let $V_{+} = \sett{v}{x(v) > 0}$, $V_{-} = \sett{v}{x(v) < 0}$, and
  denote $M = \lceil \max_{u}\card{x(u)}\rceil$.
  We partition $V_{+}$ and $V_{-}$ as
  \[
  V_{+} = \bigcup_{j=0}^{\infty} V_{+,j},
  \qquad\qquad
  V_{-} = \bigcup_{j=0}^{\infty} V_{-,j},
  \]
  where
  \begin{align*}
  &V_{+,j} = \sett{v\in V_{+}}{M(1+\eps)^{- j-1}\leq x(v)< M (1+\eps)^{- j}},\\
  &V_{-,j} = \sett{v\in V_{-}}{-M (1+\eps)^{- j}\leq x(v)< - M (1+\eps)^{- j-1}}.
  \end{align*}
  \paragraph{The construction of the vector $y$.}
  To construct the vector $y$, for each vertex $V\in V$ we find the part $V_{+,j}$ (if $x(v)$ is positive) or $V_{-,j}$ (if $x(v)$ is negative) that $v$ belongs to,
  and set the value of $y(v)$ to be one of the endpoints of the interval defining that part. Towards this end, for each $j$ we will choose points $p_{j,+}$ and $p_{j,-}$
  (in a way that we describe shortly) in the intervals
  $[M(1+\eps)^{- j-1},M(1+\eps)^{- j}]$ and $[-M(1+\eps)^{- j},-M(1+\eps)^{- j-1}]$
  respectively, and use them to define the vector $y$ that satisfies the assertion of the lemma.
  More specifically: we set $y(u) = 0$ if $x(u) = 0$; otherwise -- without loss of generality say that
  $x(u)>0$ -- we take $j$ such that $u\in V_{+,j}$, and define $y(u) = M (1+\eps)^{- j}$
  if $x(u) > p_{j,+}$ and else $y(u) = M (1+\eps)^{- j-1}$.

  \paragraph{Choosing the cut-off points $p_{j,+}$ and $p_{j,-}$.}
  Next, we describe the choice of the points $p_{j,+}$ and $p_{j,-}$. As the argument is analogous
  in both cases, we focus on $p_{j,+}$. We define the mass of $V_{+,j}$ as the cardinality of it.
  We say $V_{+,j}$ is heavy if its mass is at least $\eps n$, and otherwise we say it is light.
  The choice of $p_{+,j}$ is done differently in light intervals and in heavy intervals, and below
  we elaborate on these two cases. To do so, we consider the cost associated with interval $j$,
  \[
    {\sf cost}_j = \frac{1}{n}\sum\limits_{u\in V_{+,j}, v\text{ neighbour}}{(x(u) - x(v))^2},
  \]
  and show that the cost of $V_{+,j}$ only increases by factor $O(1/\gamma^3)$ as a result of
  the operation on $V_{+,j}$.

  \paragraph{{\bf Light intervals.}}
  In this case, we take $p_{j,+}$ to be the midpoint of its respective interval,
  that is
  \[
  p_{j,+} = \frac{1}{2}\left(M (1+\eps)^{- j-1}+M (1+\eps)^{- j}\right).
  \]
  We denote by $L_j$ the length of the interval, that is
  \[
  L_j = M (1+\eps)^{- j} - M (1+\eps)^{- j-1}
  =M\eps (1+\eps)^{- j-1}.
  \]
  Let ${\sf cost}_j$ be the contribution of points from $V_{+,j}$ to the average square distances i.e.
  \[
  {\sf cost}_j = \frac{1}{n}\sum\limits_{u\in V_{+,j}, v\text{ neighbour}}{(x(u) - x(v))^2}.
  \]
  We show that doing the rounding at $p_{j,+}$ as described above only increases ${\sf cost}_j$ by factor $O\left(\frac{1}{\gamma^3}\right)$.

  Let $\ell_1 = \frac{\gamma}{K}$, where $K$ is an absolute constant to be determined later.
  Consider the interval
  $B_1 = [p_{j,+}-\ell_1 L_j, p_{j,+}+\ell_1 L_j]$.
  We will show that ${\sf cost}_j > \frac{\gamma \ell_1^2 L_j^2}{100}d\card{B_1}$, and then noting that rounding at $p_{j,+}$
  increases ${\sf cost}_j$ additively by at most $d\card{B_1} L_j^2$ it follows that the cost of $V_{+,j}$ increases by
  at most factor $O(1/\gamma^3)$ as a result of the rounding.

  We now prove that ${\sf cost}_j > \frac{\gamma \ell_1^2 L_j^2}{100}d\card{B_1}$, and to do so we assume towards contradiction the contrary inequality holds.
  Let $\ell_2 = 2\ell_1$, and set $B_2 = [p_{j,+} - \ell_2L_j, p_{j,+} + \ell_2L_j]$.
  We claim that $\card{B_2}\geq \left(1+\frac{\gamma}{2}\right)\card{B_1}$. Indeed, otherwise we consider outgoing edges from $B_1$,
  and note that as $\card{B_1}\leq \card{V_{j,+}}\leq \eps n$, the small set expansion property implies that at least $d\gamma \card{B_1}$ of them escape outside $B_1$, where $d$
  is the degree of $G'$. As $\card{B_2}\leq \left(1+\frac{\gamma}{2}\right)\card{B_1}$, at most $\frac{d\gamma}{2} \card{B_1}$ of them can
  go to $B_2\setminus B_1$, and hence at least $\frac{d\gamma}{2}\card{B_1}$ of them go outside $B_2\setminus B_1$. Hence, we conclude that
  \[
  {\sf cost}_j \geq \frac{d\gamma}{2}\card{B_1} (\ell_2-\ell_1)^2 L_j^2 = \frac{\gamma\ell_1^2}{2}L_j^2 d\card{B_1},
  \]
  and contradiction. We thus have that $\card{B_2}\geq \left(1+\frac{\gamma}{2}\right)\card{B_1}$.

  Next, we set $\ell_3 = \ell_2 + \sqrt{\frac{\card{B_1}}{\card{B_2}}}\ell_1$ and
  consider $B_3 = [p_{j,+} - \ell_3 L_j, p_{j,+} - \ell_3 L_j]$. Similarly to before, if
  $\card{B_3}\leq \left(1+\frac{\gamma}{2}\right)\card{B_2}$ we get a contradiction to the upper bound on ${\sf cost}_j$.
  Continuing this argument, we define iteratively $\ell_{r+1} = \ell_{r} + \sqrt{\frac{B_{1}}{\card{B_r}}} \ell_1$ and
  then $B_{r+1} = [p_{j,+} - \ell_{r+1}L_j, p_{j,+} + \ell_{r+1}L_j]$, and prove that
  $\card{B_{r+1}}\geq \left(1+\frac{\gamma}{2}\right)\card{B_r}$ if $\ell_{r+1}\leq 1/10$. Note that
  \[
  \ell_{j+1}\leq \ell_1\sum\limits_{k=0}^{\infty}\left(1+\frac{\gamma}{2}\right)^{-k/2}
  =O(\ell_1/\gamma)
  =O(1/K)
  \leq 1/10
  \]
  provided $K$ is sufficiently large, so we way we keep on increasing $\card{B_r}$ which is clearly impossible.

  This contradiction implies that ${\sf cost}_j \geq \frac{\gamma \ell_1^2 L^2}{100}d\card{B_1}$, hence the cost of $V_{+,j}$ increases
  by at most factor $1/(\gamma \ell_1^2) = O(1/\gamma^3)$.
  In total, we get that throughout the process the cost of each light interval increases by factor at most $O(1/\gamma^3)$.

  Additionally, the cost of each edge (between intervals) may increase by a factor of $4$. Therefore, after rounding all light intervals,
  the cost of each light interval increases by factor at most $O(1/\gamma^3)$, and of any other interval by factor $O(1)$.

  \paragraph{{\bf Heavy intervals.}} If $V_{+,j}$ is heavy,
  say it contains between $r\eps n$ and $(r+1)\eps n$ points for $r\in\mathbb{N}$,
  we partition $J = [M(1+\eps)^{- j-1}, M (1+\eps)^{- j})$ into $3$ equal thirds, consider the middle third,
  and take equally spaced points $v_1,\ldots,v_{r+2}$ in it such that
  $v_1$ and $v_{r+2}$ are the endpoints of the middle third. In a formula,
  \[
  v_i = M(1+\eps)^{- j-1} + L_j \frac{r+1+i}{3r+5},
  \]
  where again $L_j$ is the length of $J$, $L_j = M\eps (1+\eps)^{- j-1}$.

  We could repeat the previous argument with each one of the points
  $v_i$, but this time since $V_{j,+}$ is heavy, the argument could terminate
  because small-set expansion no longer holds, in which case we have
  enlarged the interval around
  $v_i$ to contain at least $\eps n$ vertices. We change the definition of
  $\ell_1$ to be $\frac{\gamma}{r K}$, so that if the argument fails
  for $v_i$ it means we have found an interval around $v_i$ of length at most
  $L_j/(100 r)$ that contains at least $\eps n$ vertices. Note that these intervals
  must be disjoint (as the distance between two $v_i$'s is greater than $L_j/(50 r)$), and
  thus in total these interval cover no more than $\card{V_{+,j}} \leq (r+1) \eps n$
  vertices, meaning the argument could fail on at most $(r+1)$ many of the $v_i$'s.
  Thus, there is a $v_i$ in which the argument succeeds, so we choose it to be $p_{j,+}$, so that
  running the previous argument using $p_{j,+}$ gives that ${\sf cost}_j$ increases by at most factor
  $O\left(1/(\gamma \ell_1^2)\right) = O(\gamma^{-3}\eps^{-2})$.

  Besides that, the cost of edges that have one endpoint in $V_{+,j}$ and one outside $V_{+,j}$
  may increase by factor $O(1)$.
  Thus, in total the effect of rounding the heavy intervals is that it may increase the cost of each
  heavy interval by at most $O(\gamma^{-3}\eps^{-2})$, and each light interval by factor at most $O(1)$.
\end{proof}

With Lemma~\ref{lem:cheeger_quantized} in hand we can prove Theorem~\ref{thm:main_cheeger}, restated below.
\begin{thm}
  There exists an absolute constant $C>0$, such that the following holds for all $\eps,\gamma,\delta>0$.
  Suppose that $G$ is a $d$-regular graph whose second normalized eigenvalue is at least $1-\delta$,
  and suppose $G$ is an $(\eps,\gamma)$ small set expander. Then there exists a set $S$ of vertices
  of fractional size at most $n/2$, such that
  \[
  \Phi(S)\leq C\cdot \frac{\delta}{\gamma^3\eps^3}.
  \]
  Furthermore, such $S$ can be found efficiently.
\end{thm}
\begin{proof}
  Let $x$ be an eigenvector of $G$ with eigenvalue $\lambda_2\geq 1- \delta$ and suppose that
  $\Expect{}{x(v)^2} = 1$; hence $\Expect{(u,v)\in E}{(x(u) - x(v))^2}\leq 2\delta$.
  We assume without loss of generality that $x_1\leq x_2\leq\ldots\leq x_n$.

  Let $z = x+c \vec{1}$ an appropriate constant $c$ so that $z_{n/2} = 0$. We note that this constant is at most $3$ in
  absolute value, otherwise the absolute value of $x_{n/2}$ is greater than $3$, say $x_{n/2} > 3$, and then
  \[
  \Expect{v}{x(v)^2}\geq \frac{1}{2} 3^2 > 1.
  \]
  We also note that $\Expect{v}{z(v)^2} = \Expect{v}{x(v)^2} + c^2$ (we used the fact that $x$ is perpendicular to
  $\vec{1}$), which is at constant between $1$ and $10$, so we may divide $z$ by an appropriate constant
  and get its $2$-norm back to $1$. We note that performing all of these operations
  changes $\delta$ by at most a constant factor.

  Using Lemma~\ref{lem:cheeger_quantized} on $z$ we find an $\eps$-quantized $y$ satisfying the properties of the lemma.
  We divide the entries of $y$ by a suitable factor $A$ so that $y_1^2 + y_n^2 = 1$, and note that the operation of Lemma~\ref{lem:cheeger_quantized}
  preserves orders and $0$-entries so that $y_1\leq\ldots\leq y_n$ and $y_{n/2} = 0$.

  We now perform the standard Cheeger's inequality rounding scheme. Namely, choose $t$ in $[y_1,y_n]$ according to the density function $2\card{t}$. We note that
  \begin{enumerate}
    \item If $a,b$ have the same signs, the probability that $t\in [a,b]$ is $\card{b^2 - a^2}$;
    \item If $a,b$ have different signs, the probability that $t\in [a,b]$ is $a^2 + b^2$.
    \item Therefore, in either case the probability that $t\in [a,b]$ is at most $\card{a-b}(\card{a} + \card{b})$.
  \end{enumerate}
  Let $S = \sett{v}{y(v)\leq t}$. We compute the expectation of the number of edges that escape $S$,
  as well as the expectation of $\min(\card{S},\card{V\setminus S})$.

  By the above observation, the probability that a given edge $(u,v)$ crosses the cut of $S$ is
  at most $\card{y(u) - y(v)}(\card{y(u)} + \card{y(v)})$. If $y(u) = y(v)$ then this is $0$
  and in particular equal to $\card{y(u) - y(v)}^2$; if $y(u), y(v)$ are negated in signs this is same as
  $\card{y(u) - y(v)}^2$; finally, if either $y(v) = 0$ or $y(u) = 0$ this is again the same as $\card{y(u) - y(v)}^2$.
  Otherwise, as $y$ is $\eps$-quantized we have that either $\card{y(u)}\geq (1+\eps)\card{y(v)}$ or $\card{y(u)}\leq (1-\eps)\card{y(v)}$,
  and then
  \[
  \card{y(u) - y(v)}\card{y(v)}
  = \card{\frac{y(u)}{y(v)} - 1} \card{y(v)}^2
  \leq \eps^{-1} \card{\frac{y(u)}{y(v)} - 1}^2 \card{y(v)}^2
  =\eps^{-1}\card{y(u) - y(v)}^2.
  \]
  Similarly, flipping the roles of $u$ and $v$ we have that either $\card{y(v)}\geq (1+\eps)\card{y(u)}$ or $\card{y(v)}\leq (1-\eps)\card{y(u)}$, and then
  \[
  \card{y(u) - y(v)}\card{y(u)}
  = \card{\frac{y(v)}{y(u)} - 1} \card{y(u)}^2
  \leq \eps^{-1} \card{\frac{y(v)}{y(u)} - 1}^2 \card{y(u)}^2
  =\eps^{-1}\card{y(u) - y(v)}^2.
  \]

  Combining everything, we get that the probability that a given edge $(u,v)$ crosses the cut of $S$ is at most $2\eps^{-1}\card{y(u) - y(v)}^2$,
  so the expected number of edges between $S$ and $V\setminus S$ is at most
  \[
  \sum\limits_{(u,v)\in E}{2\eps^{-1}\card{y(u) - y(v)}^2}
  \leq O\left(\frac{\delta}{\gamma^3\eps^3}\frac{1}{A^2}\card{E}\right).
  \]

  As for $\Expect{}{\min({\card{S},\card{V\setminus S}})}$, if $t\leq 0$ the set $S$ is smaller and otherwise $V\setminus S$
  is smaller. We show that for each $v$, the probability that $v$ is in the smaller among $S$ and $V\setminus S$ is at least
  $y(v)^2$, from which it follows that
  \[
  \Expect{}{\min({\card{S},\card{V\setminus S}})}\geq \sum\limits_{v\in V}{y(v)^2}\geq \frac{1}{A^2}(1-\eps)n.
  \]
  This would imply Theorem~\ref{thm:main_cheeger}, as then we get that
  \[
  \frac{\Expect{}{\card{{\sf Edges}(S, V\setminus S)}}}{d\Expect{}{\min({\card{S},\card{V\setminus S}})}}
  =O\left(\frac{\delta}{\gamma^3\eps^3}\right),
  \]
  hence there is a choice of $t$ for which
  \[
  \frac{\card{{\sf Edges}(S, V\setminus S)}}{d\min({\card{S},\card{V\setminus S}})}=O\left(\frac{\delta}{\gamma^3\eps^3}\right).
  \]
  Furthermore, standard techniques show that such $t$ can be found in polynomial time, giving us a set $S$ as desired.

  To finish the proof, we argue that for any vertex $v$, the probability that $v$ is in the smaller among $S$ and $V\setminus S$ is at least
  $y(v)^2$. We consider the case that $y(v)\leq 0$ and the case that $y(v)> 0$ separately. In the first case, we have
  \[
  \Prob{}{v\text{ in smaller among $S$, $V\setminus S$}}
  =\Prob{}{t\leq 0}\cProb{}{t\leq 0}{v\in S}
  =\Prob{}{t\leq 0}\cProb{}{t\leq 0}{y(v)\leq t},
  \]
  which is equal to $\Prob{}{t\in [y(v),0]} = y(v)^2$. If $y(v)> 0$,
  \[
  \Prob{}{v\text{ in smaller among $S$, $V\setminus S$}}
  =\Prob{}{t> 0}\cProb{}{t> 0}{v\in V\setminus S}
  =\Prob{}{t> 0}\cProb{}{t\leq 0}{y(v)> t},
  \]
  which is equal to $\Prob{}{t\in [0,y(v)]} = y(v)^2$.
  \end{proof}

\subsection{Dense cuts on SSE's}
Using the same proof strategy as above, one may establish the following improved spectral partitioning result based on the smallest
eigenvalue of $G$ when the graph is a small-set expander~\cite{Trevisan}.
\begin{thm}\label{thm:dense_cut}
  There exists an absolute constant $C>0$, such that the following holds for all $\eps,\gamma,\delta>0$.
  Suppose that $G$ is a $d$-regular graph with $\lambda_n(G)\leq -1+\delta$,
  and suppose $G$ is an $(\eps,\gamma)$ small set expander. Then there is $y\in\{-1,0,1\}^n$ such that
  \[
  \frac{\sum\limits_{(u,v)\in E}\card{y(u)+y(v)}}{d\sum\limits_{u\in V}{\card{y(u)}}}\leq C\cdot \frac{\delta}{\gamma^3\eps^3}.
  \]
  Furthermore, such $S$ can be found efficiently.
\end{thm}
As the proof is very close in spirit, we omit the details. We remark that Trevisan~\cite{Trevisan} proves a version of Theorem~\ref{thm:dense_cut}
for general graphs and uses it to get a spectral approximation algorithm for Max-cut with better approximation ratio than $1/2$.
In that context, a solution as in Theorem~\ref{thm:dense_cut} is to be interpreted as a partial cut, and after finding $y$ Trevisan considers
the set of vertices which receive value $0$ in $y$ and applies the algorithm on them recursively. One may expect such approach to also work in our case to get an algorithm for Max-cut,
however the recursive nature of Trevisan's approach may not preserve the small-set expansion of the induced graph. We were therefore led to seeking a more direct
approach for applying the soft-rounding idea to the Max-cut problem, and indeed our proof of Theorem~\ref{thm:main_MC} was found in this way.

\section{Strong Parallel Repetition for Unique-Games on SSE's}
In this section, we prove a more general version of Theorem~\ref{thm:main_UG} for the class of projection games.
Our proof goes through the information theoretic approach to parallel repetition~\cite{Raz,Holenstein,Rao,BravermanGarg}, and
we will closely follow the argument in~\cite{BravermanGarg}. Our argument differs only towards the end of the argument, but
we sketch for completeness. We will prove that for sufficiently large absolute constant $K>0$,
setting $T = \lceil{\frac{K}{\eps^3\gamma^3\delta}\rceil}$, for $t\geq T$ it holds that
${\sf val}(\Psi^{t})\leq \left(1-\Omega(\delta \eps^3\gamma^3)\right)^t$.
The result for any other $t$ then follows automatically: for any $t\leq T$, letting $\ell = {\lfloor T/t\rfloor}$
we have that $T\geq \ell t$ so
\[
1-\Omega(1)\geq {\sf val}(\Psi^{T})\geq {\sf val}(\Psi^{\ell t}) \geq {\sf val}(\Psi^{t})^{\ell},
\]
so
\[
{\sf val}(\Psi^{t})\leq (1-\Omega(1))^{1/\ell} = e^{-\Omega(t/T)} = (e^{-\Omega(1/T)})^t = (1-\Omega(1/T))^{t},
\]
as desired. From now on, we assume that $t\geq T$.

\paragraph{Notations.} We will denote random variables by boldface capital letters, and instantiations of
them by small letters. We use standard information theoretic tools presented in the appendix for completeness
(see also~\cite{elements}). The KL-divergence from $Q$ to $P$ is denoted by $\DKL{P}{Q}$, however for space considerations, we will sometimes denote
it by
\[
\mathrm{D}_{\text{KL}}\begin{pmatrix}
    P\\
    \vspace{-2ex}\\
    \overline{\overline{\phantom{P}}}
    \vspace{-2ex}\\
    Q
    \end{pmatrix}.
\]

\subsection{Parallel Repetition: the Information Theoretic Set-up}
We will refer to the players in the game as Alice and Bob. The challenges of Alice are in
the $t$-fold repeated game $\Psi^{\otimes n}$ are denoted by ${\bf X}_1,\ldots, {\bf X}_t$ and her
answers are denoted by $({\bf A}_1,\ldots,{\bf A}_t) = f({\bf X}_1,\ldots,{\bf X}_t)$. Similarly, the challenges
of Bob are denoted by ${\bf Y}_1,\ldots,{\bf Y}_t$, and his answers are denoted by
$({\bf B}_1,\ldots,{\bf B}_t) = g({\bf Y}_1,\ldots,{\bf Y}_t)$. Finally, we denote by $W$ the event that Alice and Bob
win on all of the $t$ challenges. We will prove the counter positive statement, namely that if
$\Prob{}{W} > (1-c\delta\eps^3\gamma^3)^t$ for sufficiently small absolute constant $c$, then ${\sf val}(\Psi) > 1-\delta$
(we recall that $t\geq T$).
\subsubsection{The information theoretic approach}
Let ${\bf s_g}, {\bf s_h}$ be uniformly chosen integers in $\{3/4 t+1,\ldots,t\}$, and let $\bm{\sigma}\in S_t$ be a uniformly chosen permutation.
Denote
\[
{\bf H} = \bm{\sigma}([{\bf s_h}]) = \sett{\bm{\sigma}(i)}{i\in [{\bf s_h}]}, \qquad
{\bf G} = \sigma(\{t-{\bf s_g}+1,\ldots,t\}) = \sett{\sigma(i)}{i\in\{t-{\bf s_g}+1,\ldots,t\}}.
\]
Let ${\bf I}\in {\bf G}\cap {\bf H}$ be chosen uniformly, and let $\bm{\ell}\in [t/4]$ be chosen uniformly.
Let ${\bf S}\subseteq ({\bf G}\cap {\bf H})\setminus\{{\bf I}\}$ be chosen uniformly of size $\bm{\ell}$. Denote
\[
{\bf L}_{{\bf S},{\bf G},{\bf H},{\bf I}} = ({\bf X}_{{\bf G}\setminus\{{\bf I}\}}, {\bf Y}_{{\bf H}\setminus\{{\bf I}\}}, {\bf B}_{\bf S}).
\]

Below we state Lemmas 5.2, 5.3, 5.4 and 5.5 from~\cite{BravermanGarg}. We omit the proofs as they are identical to the proofs therein.
\begin{claim}\label{claim:MI1}
  $\Expect{{\bf S},{\bf G},{\bf H},{\bf I}}{\MI({\bf A}_{\bf I}; {\bf Y}_{\bf I} | {\bf X}_{\bf I}, {\bf L}_{{\bf S},{\bf G},{\bf H},{\bf I}}, W)}\leq \frac{4}{t}\log(1/\Prob{}{W})$.
\end{claim}

\begin{claim}\label{claim:MI2}
  $\Expect{{\bf S},{\bf G},{\bf H},{\bf I}}{\MI({\bf L}_{{\bf S},{\bf G},{\bf H},{\bf I}} ; {\bf Y}_{\bf I} | {\bf X}_{\bf I}, W)}\leq \frac{8}{t}\log(1/\Prob{}{W})$.
\end{claim}

\begin{claim}\label{claim:MI3}
  $\Expect{{\bf S},{\bf G},{\bf H},{\bf I}}{\MI({\bf L}_{{\bf S},{\bf G},{\bf H},{\bf I}} ; {\bf X}_{\bf I} | {\bf Y}_{\bf I}, W)}\leq \frac{8}{t}\log(1/\Prob{}{W})$.
\end{claim}

\begin{claim}\label{claim:MI4}
  $\Expect{{\bf S},{\bf G},{\bf H}, {\bf I}}{\MI({\bf B}_{\bf I} ; 1_W | {\bf X}_{\bf I}, {\bf Y}_{\bf I}, {\bf L}_{{\bf S},{\bf G},{\bf H},{\bf I}}, {\bf A}_{\bf I})}
  \leq \frac{4}{t}H(1_W)$.
\end{claim}

The next claim is Lemma 4.6 from~\cite{BravermanGarg}.
\begin{claim}\label{claim:MI5}
  $\Expect{{\bf I}}{
  \DKL{{\bf X}_{\bf I}, {\bf Y}_{\bf I}~|~W}{{\bf X}_{\bf I}, {\bf Y}_{\bf I}}}
  \leq \frac{1}{t}\log\left(\frac{1}{\Prob{}{W}}\right)$.
\end{claim}

The next claim is Lemma 4.9 from~\cite{BravermanGarg}.
\begin{claim}\label{claim:break_dependency}
  Suppose $G,H,S_a,S_b\subseteq [t]$ and $i\in [t]$ such that $G\cup H = [t]\setminus\set{i}$.
  Then for any $\bar{x}, \bar{y}, \bar{a}, \bar{b}, x, y$ in the support
  of ${\bf X}_G, {\bf Y}_H, {\bf A}_{S_a}, {\bf B}_{S_b}, {\bf X}_i$ and ${\bf Y}_i$,
  for all $a,b$ it holds that
  \begin{align*}
  &\cProb{}
  {({\bf X}_G, {\bf Y}_H, {\bf A}_{S_a}, {\bf B}_{S_b}, {\bf X}_i, {\bf Y}_i) = (\bar{x}, \bar{y}, \bar{a}, \bar{b}, x, y)}
  {{\bf A}_i = a, {\bf B}_i = b}\\
  &
  =\cProb{}
  {({\bf X}_G, {\bf Y}_H, {\bf A}_{S_a}, {\bf X}_i) = (\bar{x}, \bar{y}, \bar{a}, x)}
  {{\bf A}_i = a}
  \cdot
  \cProb{}
  {({\bf X}_G, {\bf Y}_H, {\bf B}_{S_b}, {\bf Y}_i) = (\bar{x}, \bar{y}, \bar{b}, x)}
  {{\bf B}_i = b}.
  \end{align*}
\end{claim}

Next, we state and prove an analog to Lemma 5.6 from~\cite{BravermanGarg}. The Lemma therein applies to the case that
the probability of $W$ is small, and below we make a small adjustment to it (using the fact that the number of repetitions
is assumed to be large in our case) to remove that assumption.
\begin{lemma}\label{lemma:final_BG}
  Suppose $\zeta = \delta\eps^3\gamma^3$ and that $\Prob{}{W}\geq 2^{-\zeta t/K}$.
  There exist fixings ${\bf G} = G$, ${\bf H} = H$, ${\bf S} = S$ and ${\bf I} = I$ such that the following properties hold:
  \begin{enumerate}
    \item
    \[
    \Expect{(x,y)\sim ({\bf X}_I, {\bf Y}_I)|W}{\DKL{{\bf L}_{S,G,H,I}~|~{\bf X}_{I} = x, {\bf Y}_I = y, W}{{\bf L}_{S,G,H,I}~|~{\bf X}_{I} = x, W}}\leq
    O\left(\frac{\zeta}{K}\right).
    \]

    \item
    \[
    \Expect{(x,y)\sim ({\bf X}_I, {\bf Y}_I)|W}{\DKL{{\bf L}_{S,G,H,I}~|~{\bf X}_{I} = x, {\bf Y}_I = y, W}{{\bf L}_{S,G,H,I}~|~{\bf Y}_{I} = y, W}}\leq O\left(\frac{\zeta}{K}\right).
    \]

    \item
    \[
    \DKL{{\bf X}_I, {\bf Y}_I~|~W}{{\bf X}_I, {\bf Y}_I}\leq
    O\left(\frac{\zeta}{K}\right).
    \]

    \item
    \begin{align*}
    \hspace{-7ex}\E_{\substack{(x,y)\sim ({\bf X}_I, {\bf Y}_I)|W \\ L\sim {\bf L}_{S,G,H,I}|{\bf X}_{I} = x, {\bf Y}_I = y, W}}
    \Big[&\DKL{{\bf A}_I| {\bf X}_{I} = x, {\bf Y}_I = y, {\bf L}_{S,G,H,I} = L, W}{
    {\bf A}_I| {\bf X}_{I} = x, {\bf L}_{S,G,H,I} = L, W}\Big]\\
    &\leq
    O\left(\frac{\zeta}{K}\right).
    \end{align*}

    \item
    \begin{align*}
    \hspace{-3ex}\E_{\substack{(x,y)\sim ({\bf X}_I, {\bf Y}_I)|W \\ L\sim {\bf L}_{S,G,H,I}|X_{I} = x, Y_I = y, W}}
    &\left[
    \mathrm{D}_{\text{KL}}
    \begin{pmatrix}
    {\bf A}_I,{\bf B}_I| {\bf X}_{I} = x, {\bf Y}_I = y, {\bf L}_{S,G,H,I} = L, W\\
    \vspace{-2ex}\\
    \overline{\overline{\phantom{{\bf A}_I| {\bf X}_{I} = x, {\bf L}_{S,G,H,I} = L, W\otimes {\bf B}_I| {\bf Y}_{I} = y, {\bf L}_{S,G,H,I} = L, W}}}
    \vspace{-2ex}\\
    {\bf A}_I| {\bf X}_{I} = x, {\bf L}_{S,G,H,I} = L, W\otimes {\bf B}_I| {\bf Y}_{I} = y, {\bf L}_{S,G,H,I} = L, W
    \end{pmatrix}\right]\\
    &~\leq O\left(\frac{\zeta}{K}\right).
    \end{align*}
  \end{enumerate}
\end{lemma}
\begin{proof}
  We calculate the expected value of each one of these over the choice of ${\sf S}, {\bf G}, {\bf H}, {\bf I}$, show that each one of these expectations
  is at most $O\left(\frac{1}{t}\log\left(\frac{1}{\Prob{}{W}}\right)\right)$, and then the result follows from Markov's inequality together with the
  union bound.

  The expectation of the first item is, by Fact~\ref{fact:mutual_div_KL}, equal to the mutual information in Claim~\ref{claim:MI2}, so the bound follows
  from there. The expectation of the second item is, by Fact~\ref{fact:mutual_div_KL}, equal to the mutual information in Claim~\ref{claim:MI3}, so the bound
  follows from there. The expectation of the third item is computed in Claim~\ref{claim:MI5}. The expectation of the fourth item is equal to the mutual information in Claim~\ref{claim:MI1}, and the bound follows from there.

  For the fifth expectation, using the chain rule we can write it as
  \begin{align*}
  &\E_{\substack{{\bf G}, {\bf H}, {\bf S}, {\bf I}\\ (x,y)\sim ({\bf X}_{\bf I}, {\bf Y}_{\bf I})|W \\ L\sim {\bf L}_{{\bf S},{\bf G},{\bf H},{\bf I}}|
  X_{{\bf I}} = x, Y_{\bf I} = y, W}}
  \Big[
  \underbrace{\DKL{{\bf A}_{\bf I}| X_{{\bf I}} = x, Y_{\bf I} = y, {\bf L}_{{\bf S},{\bf G},{\bf H},{\bf I}} = L, W}{
    {\bf A}_{\bf I}| X_{{\bf I}} = x, {\bf L}_{{\bf S},{\bf G},{\bf H},{\bf I}} = L, W}}_{(\rom{1})}\\
  &+\hspace{-3ex}
  \underbrace{\E_{a\sim {\bf A}_{\bf I}\left|\substack{X_{{\bf I}} = x, Y_{\bf I} = y\\ {\bf L}_{{\bf S},{\bf G},{\bf H},{\bf I}} = L, W}\right.}
  \DKL{{\bf B}_{\bf I}| X_{{\bf I}} = x, {\bf Y}_I = y, {\bf L}_{{\bf S},{\bf G},{\bf H},{\bf I}} = L,
  {\bf A}_{\bf I} = a, W}
  {{\bf B}_{\bf I}| {\bf Y}_{\bf I} = y, {\bf L}_{{\bf S},{\bf G},{\bf H},{\bf I}} = L}}_{(\rom{2})}\Big]
  \end{align*}
  The expectation of $(\rom{1})$ is the expectation of the fourth item in the lemma, hence at most
  $O\left(\frac{1}{t}\log\left(\frac{1}{\Prob{}{W}}\right)\right)$. For $(\rom{2})$, first looking at Claim~\ref{claim:break_dependency},
  summing over $a$ and setting $S_a = \{{\bf I}\}$, $S_b = {\bf S}$ we get that the distributions
  ${\bf B}_{\bf I}| {\bf Y}_{\bf I} = y, {\bf L}_{{\bf S},{\bf G},{\bf H},{\bf I}} = L$
  and
  ${\bf B}_{\bf I}| {\bf X}_{\bf I} = x, {\bf Y}_{\bf I} = y, {\bf L}_{{\bf S},{\bf G},{\bf H},{\bf I}} = L, {\bf A}_{{\bf I}} = a$
  are identical, so our goal is to upper bound the expectation of
  \[
  \DKL{{\bf B}_{\bf I}| {\bf X}_{{\bf I}} = x, {\bf Y}_{\bf I} = y, {\bf L}_{{\bf S},{\bf G},{\bf H},{\bf I}} = L, W,
  {\bf A}_{\bf I} = a, W}
  {{\bf B}_{\bf I}| {\bf X}_{\bf I} = x, {\bf Y}_{\bf I} = y, {\bf L}_{{\bf S},{\bf G},{\bf H},{\bf I}} = L, {\bf A}_{{\bf I}} = a}.
  \]
  Letting $w\sim 1_W$ and then sampling the rest of the random variables conditioned on it, we see that the expectation
  of the above expression is at most $\frac{1}{\Prob{}{W}}$ times the expectation of
    \[
  \mathrm{D}_{\text{KL}}
    \begin{pmatrix}
   {\bf B}_{\bf I}| X_{{\bf I}} = x, {\bf Y}_I = y, {\bf L}_{{\bf S},{\bf G},{\bf H},{\bf I}} = L, 1_W=w,
  {\bf A}_{\bf I} = a\\
    \vspace{-2ex}\\
    \overline{\overline{\phantom{{\bf B}_{\bf I}| X_{{\bf I}} = x, Y_I = y, {\bf L}_{{\bf S},{\bf G},{\bf H},{\bf I}} = L, 1_W=w,
  {\bf A}_{\bf I} = a}}}
    \vspace{-2ex}\\
    {{\bf B}_{\bf I}| {\bf X}_{\bf I} = x, {\bf Y}_{\bf I} = y, {\bf L}_{{\bf S},{\bf G},{\bf H},{\bf I}} = L, {\bf A}_{{\bf I}} = a}
    \end{pmatrix}
  \]
  since the probability that $w = 1$ is $\Prob{}{W}$. By Fact~\ref{fact:mutual_div_KL}, the expectation of the last KL-divergence
  is the mutual information in Claim~\ref{claim:MI4}, hence at most $\frac{4}{t} H[1_W]$.

  All in all, we get that the expectation of $(\rom{2})$ is upper bounded by
  \[
  \frac{4 H[1_W]}{t\Prob{}{W}}.
  \]
  Denote $q = \Prob{}{W}$. If $q\leq 1/2$, then we may bound
  $H[1_W]\leq O(q\log(1/q))$, and then we get the bound $O\left(\frac{\log(1/\Prob{}{W})}{t}\right) \leq O\left(\frac{\zeta}{K}\right)$
  on $(\rom{2})$. Otherwise, $q > 1/2$ and we have the bound
  $H[1_W]\leq 1$, and we get the bound $O(1/t)$. By the assumption in the beginning of this section, $t\geq T\geq \Omega(K/\zeta)$,
  so $O(1/t)\leq O\left(\frac{\zeta}{K}\right)$.
\end{proof}

\subsubsection{Departing from~\cite{BravermanGarg}}
In the next part of the proof, we insert an additional ingredient on top of~\cite{BravermanGarg},
and before that we quickly explain how the proof there proceeds using Lemma~\ref{lemma:final_BG}.

First, a protocol is designed so that if the players received challenges from
the distribution ${\bf X}_I, {\bf Y}_I ~ | W$, then the players succeed with probability close to
$1$; given that, the third item in Lemma~\ref{lemma:final_BG} shows that the players succeed with
probability close to $1$ given challenges distributed as ${\bf X}_I, {\bf Y}_I$.

The correctness of the protocol is argued by appealing to Pinsker's inequality on
the third item of Lemma~\ref{lemma:final_BG}, to get that Alice and Bob can jointly
sample from distributions that are $O(\sqrt{\zeta/K})$ close to ${\bf L}_{S,G,H,I}~|~X_{I} = x, Y_I = y, W$,
i.e. with probability $1-O(\sqrt{\zeta/K})$ they get a joint sample from that distribution.
Assuming the joint sampling was successful, the players can sample the answer $A_I, B_I$ conditioned on the information they have so far.
Note that in the distribution ${\bf A}_I,{\bf B}_I| X_{i} = x, Y_I = y, {\bf L}_{S,G,H,I} = L, W$ the players
win the coordinate $I$ with probability $1$ (since $W$ was conditioned on), and by the fifth item in
Lemma~\ref{lemma:final_BG} the KL-divergence between that distribution and the joint distribution
of the answers of the players in the designed protocol is at most $O(\zeta/K)$, and one conclude
that the winning probability of the players in the protocol (conditioned on the joint sampling being successful)
is at least $2^{-O(\zeta/K)}\geq 1-O(\zeta/K)$. Overall, the designed protocol wins with probability
$1-O(\sqrt{\zeta/K})$.

The source of the square loss above is thus due to the application of Pinsker's inequality, and we will circumvent that
by appealing to the small set expansion property. More precisely, we will use the first and second item in Lemma~\ref{lemma:final_BG}
in order to come up with distributions $\tilde{{\bf L}}_{S,G,H,I}~|~X_{I} = x, W$, $\tilde{{\bf L}}_{S,G,H,I}~|~X_{I} = y, W$
that are very close to the original distribution (the probability of each atom changes by at most factor $(1\pm \eps)$, but from
which the players can jointly sample without losing the square root in Pinsker's inequality (and instead loses
some factors depending on the small set expansion parameters of the graph). From that point, the rest of the proof proceeds
in the same way.
\subsection{Qunatizing the Random Variables}
Throughout this section, we have a graph $G = (V,E)$, and we associate a distribution $\mathcal{D}_v$ with each
vertex $v$. We will also need to consider pseudo-distributions, which we define next.
\begin{definition}
  A pseudo-distribution $\mathcal{D}$ over a finite domain $\Omega$ is a map
  $\mathcal{D}\colon \Omega\to[0,\infty)$.
\end{definition}
Sometimes we will want to sample from a pseudo-distribution; by that, we mean that we first normalize
$\mathcal{D}$ so that the sum of its values is $1$, and then sample from it.
One may define Hellinger distance as well as statistical distance for pseudo-distributions as well
(though one has to be careful with using properties of them that only hold for distributions).

\begin{definition}
  A collection of pseudo-distributions $\mathcal{D}_{v\in V}$ over a finite domain $\Omega$ is called $\eps$-quantized if for
  any $w\in \Omega$ and $u,v\in V$ we either have that $\mathcal{D}_v(w) = \mathcal{D}_u(w)$, or
  $\mathcal{D}_v(w)\leq (1-\eps)\mathcal{D}_u(w)$, or $\mathcal{D}_v(w)\geq (1+\eps)\mathcal{D}_u(w)$.
\end{definition}

The following lemma explains the benefit of quantized pseudo-distributions: it allows us to move from Hellinger distance to
statistical distance without square root loss.
\begin{lemma}\label{lem:from_hell_to_sd}
  If $P,Q$ are pseudo-distributions such that the collection $\{P,Q\}$ is $\eps$-quantized.
  Then ${\sf SD}(P,Q) = O\left(\frac{1}{\eps}{\sf Hellinger}(P,Q)^2\right)$.
\end{lemma}
\begin{proof}
  By definitions, the inequality we wish to show is that for large enough absolute constant $C>0$ it holds that
  \[
  \sum\limits_{w\in \Omega}\card{P(w) - Q(w)}\leq \frac{C}{\eps}\sum\limits_{w\in \Omega}\card{\sqrt{P(w)} - \sqrt{Q(w)}}^2.
  \]
  We show that the inequality in fact holds term by term. Without loss of generality, $P(w)\geq Q(w)$.
  If $Q(w) = 0$ it is clear. Otherwise, letting $t = P(w)/Q(w)$, the inequality we wish to show is that
  $t-1\leq \frac{C}{\eps}(\sqrt{t} - 1)^2$. If $t=1$ the inequality is clear, and otherwise it is equivalent to
  $\sqrt{t}+1\leq \frac{C}{\eps}(\sqrt{t}-1)$.
  Solving for $t$, we get that
  \[
  \sqrt{t}\geq \frac{C/\eps+1}{C/\eps-1} = 1+\frac{2}{C/\eps-1}=1+\frac{2\eps}{C - \eps}
  \]
  By the fact that $\{P,Q\}$ is $\eps$-quantized and $t\neq 1$, it follows that $t\geq 1+\eps$ and hence $\sqrt{t}\geq 1+\eps/4$,
  and for large enough $C$ ($C=10$ will do) we get that this is indeed more than $1+\frac{2\eps}{C-\eps}$.
\end{proof}

The next lemma shows that one may modify distributions associated with vertices of small set expanders
by small multiplicative modifications to make them quantized, and only incur a constant multiplicative increase in Hellinger distance.
\begin{lemma}\label{lem:find_pseudo_quantized}
  There exists an absolute constant $C>0$, such that the following holds.
  Suppose that $G = (V,E)$ is a regular graph as well as a $(\eps,\gamma)$ small set expander, and $(\mathcal{D}_v)_{v\in V}$ is
  a collection of distributions over domain $\Omega$ such that
  \[
  \Expect{(u,v)\in E}{{\sf Hellinger}(\mathcal{D}_u,\mathcal{D}_v)^2}\leq \eta.
  \]
  Then there exists a collection of pseudo-distributions $(\widetilde{\mathcal{D}}_v)_{v\in V}$ over $\Omega$
  satisfying the following properties:
  \begin{enumerate}
    \item $\Expect{(u,v)\in E}{{\sf Hellinger}(\tilde{\mathcal{D}}_u,\tilde{\mathcal{D}}_v)^2}\leq C\frac{\eta}{\gamma^3\eps^2}$.
    \item For all $w\in \Omega$ and $v\in V$ it holds that $(1-\eps)\mathcal{D}_v(w)\leq \tilde{\mathcal{D}}_v(w)\leq (1+\eps)\mathcal{D}_v(w)$.
    \item The collection $(\tilde{\mathcal{D}}_v)_{v\in V}$ is $\eps/2$-quantized.
  \end{enumerate}
\end{lemma}
\begin{proof}
  For each $w\in \Omega$, define the contribution of $w$ to the Hellinger distance squared as
  \[
  c_w = \Expect{(u,v)\in E}{\left(\sqrt{\mathcal{D}_u(w)} - \sqrt{\mathcal{D}_v(w)}\right)^2},
  \]
  and define the vector $x_w\in [0,1]^{V}$ by $x_w(v) = \sqrt{\mathcal{D}_v(w)}$. Then
  by Lemma~\ref{lem:cheeger_quantized} we may find a vector $y_w$ that is $\eps$-quantized
  and
  \[
  c'_w \defeq \Expect{(u,v)\in E}{(y_w(u) - y_w(v))^2}\leq C\frac{c_w}{\gamma^3\eps^2}.
  \]

  For each $v$, define the pseudo-distribution $\tilde{\mathcal{D}}_v(w) = y_w(v)^2$. Then as
  $y_w$ is $\eps$-quantized, it is easy to see that this collection of distributions is $\eps/2$ quantized.
  Also, the second item follows immediately from the corresponding properties of $y_w$. Finally,
  \[
  \Expect{(u,v)\in E}{{\sf Hellinger}(\tilde{\mathcal{D}}_u,\tilde{\mathcal{D}}_v)^2}
  =\sum\limits_{w}\Expect{(u,v)\in E}{(y_w(u) - y_w(v))^2}
  \leq
  \sum\limits_{w}C\frac{c_w}{\gamma^3\eps^2}
  =\frac{C}{\gamma^3\eps^2}\sum\limits_{w}c_w
  \leq C\frac{\eta}{\gamma^3\eps^2},
  \]
  as the sum of $c_w$ is exactly $\Expect{(u,v)\in E}{{\sf Hellinger}(\tilde{\mathcal{D}}_u,\tilde{\mathcal{D}}_v)^2}$.
\end{proof}

\subsection{Using the Random Variables to Derive a Strategy}
We now have all of the ingredients to prove Theorem~\ref{thm:main_UG}. We will show that assuming that
$\Prob{}{W}\geq 2^{-\zeta t/K}$ for $\zeta = \delta\eps^3\gamma^3$, we can design a protocol for the players to win the
game with probability $1-O(\delta/K)$, so taking $K$ large enough yields a contradiction.

Apply Lemma~\ref{lemma:final_BG} to find a fixing of $S, H, G, I$. In the protocol
we design, the input of the players will be a pair of challenges distribution $(x,y)\sim ({\bf X}_I, {\bf Y}_I)$.
We will analyze the protocol under a different distribution of challenges, namely under $(x,y)\sim ({\bf X}_I, {\bf Y}_I)~|~W$,
and prove that the players succeed with probability at least $1-O(\delta/K)$. From this, Lemma~\ref{lem:move_high_prob_events}
together with the third item in Lemma~\ref{lemma:final_BG} gives that the players win under the original distribution of challenges
with probability at least $1-O(\delta/K)$, as $\zeta\leq \delta$.

We set up some notations. Given a challenge $x$ to Alice and $y$ to Bob (which are vertices in $\Psi$),
they define the random variable
\[
\mathcal{D}_x = {\bf L}_{S,G,H,I}~|~{\bf X}_I = x, W
\qquad\qquad
\mathcal{D}_y = {\bf L}_{S,G,H,I}~|~{\bf Y}_I = y, W
\]
respectively. We also define $\mathcal{D}_{x,y} = {\bf L}_{S,G,H,I}~|~{\bf X}_I = x, {\bf X}_I = y, W$ for convenience
(though we stress no player has access to it). Then the first and second items in Lemma~\ref{lemma:final_BG},
\[
\Expect{(x,y)\sim ({\bf X}_I, {\bf Y}_I)~|~W}{\DKL{\mathcal{D}_{x,y}}{\mathcal{D}_x} + \DKL{\mathcal{D}_{x,y}}{\mathcal{D}_{y}}}\leq O(\zeta/K),
\]
and applying Lemma~\ref{lem:move_KL_hell} we get that
\[
\Expect{(x,y)\sim ({\bf X}_I, {\bf Y}_I)~|~W}{{\sf Hellinger}(\mathcal{D}_x,\mathcal{D}_y)^2} \leq O(\zeta/K).
\]
Thus, by Lemma~\ref{lem:trivial_move_expect2}
\[
\Expect{(x,y)\sim ({\bf X}_I, {\bf Y}_I)}{{\sf Hellinger}(\mathcal{D}_x,\mathcal{D}_y)^2}\leq O(\zeta/K),
\]
and from Lemma~\ref{lem:find_pseudo_quantized} there are pseudo-distributions $(\widetilde{\mathcal{D}}_x)$ as in
the lemma such that
\[
\Expect{(x,y)\sim ({\bf X}_I, {\bf Y}_I)}{{\sf Hellinger}(\widetilde{\mathcal{D}}_x,\widetilde{\mathcal{D}}_y)^2}\leq
O\left(\frac{\zeta}{K\eps^2\gamma^3}\right) = O\left(\frac{\delta\eps}{K}\right).
\]
We note that ${\sf Hellinger}(\widetilde{\mathcal{D}}_x,\widetilde{\mathcal{D}}_y)^2 \leq 2$ for all $x,y$ (it can be more than
$1$ as these are pseudo-distributions, but by the second item in Lemma~\ref{lem:find_pseudo_quantized} their sum of values is
at most $(1+\eps)$). Thus we may apply Lemma~\ref{lem:trivial_move_expect2} in the other direction to conclude that
\[
\Expect{(x,y)\sim ({\bf X}_I, {\bf Y}_I)~|~W}{{\sf Hellinger}(\widetilde{\mathcal{D}}_x,\widetilde{\mathcal{D}}_y)^2}\leq O\left(\frac{\delta\eps}{K}\right).
\]
Thus, from Lemma~\ref{lem:from_hell_to_sd} we have that
\[
\Expect{(x,y)\sim ({\bf X}_I, {\bf Y}_I)~|~W}{{\sf SD}(\widetilde{\mathcal{D}}_x,\widetilde{\mathcal{D}}_y)}\leq O\left(\frac{\delta}{K}\right)
\]
From Lemma~\ref{lemma:correlated sampling}, it thus follows that the players can jointly sample $d_x\sim \widetilde{\mathcal{D}}_x$,
$d_y\sim \widetilde{\mathcal{D}}_y$ such that
\[
\Prob{\substack{(x,y)\sim ({\bf X}_I, {\bf Y}_I)~|~W \\ d_x,d_y}}{d_x\neq d_y} = O(\delta/K).
\]

We can now state the protocol for the players.

\noindent{{\bf Input}:} Alice is given $x$, Bob is given $y$ such that $(x,y)\sim ({\bf X}_I, {\bf Y}_I)~|~W$.

\noindent{{\bf Protocol:}}
\begin{enumerate}
  \item Alice and Bob use correlated sampling to jointly sample $d_x\sim \widetilde{\mathcal{D}}_x$, $d_y\sim \widetilde{\mathcal{D}}_y$.
  \item Alice samples $a\sim {\bf A}_I| {\bf X}_{I} = x, \mathcal{D}_x = d_x, W$ and sends it to the referee.
  \item Bob samples $b\sim {\bf B}_I| {\bf Y}_{I} = x, \mathcal{D}_y = d_y, W$ and sends it to the referee.
\end{enumerate}

The following claim finishes the proof.
\begin{claim}
  The above protocol is a strategy for the players that succeeds with probability at least $1-O(\delta/K)$.
\end{claim}
\begin{proof}
  Consider the alternative protocol, where instead of the first step, Alice samples $d_x$ and then comminutes it to Bob.
  Clearly, as the probability that $d_x\neq d_y$ is at most $O(\zeta/K)$, this change increases the success probability
  of the players by at most $O(\zeta/K)$, so it suffices to show that this modified protocol has success probability at
  least $1-O(\delta/K)$.

  We first analyze the protocol where we modify the distribution of $d_x$, and then argue the implication to the protocol
  with the correct distribution over $d_x$. Suppose Alice sampled $d_x\sim \mathcal{D}_x$; we show that in this case,
  the success probability is at least $1-O(\zeta/K)$. To show that, by the first item of Lemma~\ref{lemma:final_BG} and
  Lemma~\ref{lem:move_high_prob_events}, it suffices to sow that the success probability of the protocol is at least
  $1-O(\zeta/K)$ if we sample $d_x\sim \mathcal{D}_{x,y}$. In that case, the answers of the players are distributed
  according to $({\bf A}_I| {\bf X}_{I} = x, \mathcal{D}_{x,y} = d_x, W)\otimes ({\bf B}_I| {\bf Y}_{I} = x, \mathcal{D}_y = d_y, W)$.
  Thus, by the fifth item in Lemma~\ref{lemma:final_BG} and Lemma~\ref{lem:move_high_prob_events} it suffices to
  show that the success probability of the players with $(a,b)\sim ({\bf A}_I,{\bf B}_I| {\bf X}_{I} = x, {\bf Y}_I = y, \mathcal{D}_{x,y} = d_x, W)$
  is at least $1-O(\zeta/K)$, which is clear since we have conditioned on the event $W$ so all coordinates are won.

  Thus, letting $Z(d)$ be a random variable denoting the failure probability of the modified protocol with $d_x = d$,
  we have that $\Expect{x,y}{\Expect{, d\sim \mathcal{D}_x}{Z(d)}}\leq O(\zeta/K)$, hence
  \[
  \Expect{x,y}{
  \Expect{d\sim \widetilde{\mathcal{D}}_x}{Z(d)}}
  =\Expect{x,y}{
  \sum_d\Prob{}{\widetilde{\mathcal{D}}_x=d} Z(d)}
  \leq
  \Expect{x,y}{(1+\eps)
  \sum_d\Prob{}{\mathcal{D}_x=d} Z(d)}
  =O\left(\frac{\zeta}{K}\right).
  \]
  Therefore the success probability of the protocol when $d_x\sim \widetilde{\mathcal{D}}_x$ is at least
  $1-O\left(\frac{\zeta}{K}\right)$, finishing the proof.
\end{proof}

\bibliographystyle{plain}
\bibliography{ref}

\appendix
\section{Information Theory}
In this section, we present a few basic notions from information theory. First, we define the notions of Shannon entropies and conditional Shannon entropies.
\begin{definition}
  Let $X,Y$ be random variables with a finite support.
  \begin{enumerate}
    \item The Shannon entropy of $X$
  is $\HH[X] = \sum\limits_{x}{\Prob{}{X = x}\log\left(\frac{1}{\Prob{}{X = x}}\right)}$.
    \item The Shannon entropy of $X$ conditioned on $Y$ is
    $\HH[X~|~Y] = \Expect{y\sim Y}{\HH[X~|~Y = y]}$, where
    $\HH[X~|~Y = y] = \sum\limits_{x}{\cProb{}{Y=y}{X = x}\log\left(\frac{1}{\cProb{}{Y=y}{X = x}}\right)}$.
  \end{enumerate}
\end{definition}

Second, we define mutual information between random variables as well as conditional mutual information.
\begin{definition}
Let $X,Y,Z$ be random variables with a finite support.
  \begin{enumerate}
    \item The mutual information between $X$ and $Y$ is
    $\MI[X;Y] = \HH[X] - \HH[X|Y]$.
    \item The mutual information between $X,Y$ conditioned on $Z$ is
    $\MI[X;Y~|~Z] = \HH[X~|~Z] - \HH[X~|~Y,Z]$.
  \end{enumerate}
\end{definition}

Third, we define the KL-divergence between random variables.
\begin{definition}
Let $X,Y$ be random variables with a finite support. The KL-divergence from $Y$ to $X$ is
    $\DKL{X}{Y} = \sum\limits_{x,y} \Prob{}{X = x}\log\left(\frac{\Prob{}{X = x}}{\Prob{}{Y = y}}\right)$.
\end{definition}

We will need the following standard facts from information theory (for proofs, see~\cite{elements} for example).
\begin{fact}\label{fact:mutual_div_KL}
  Let $X,Y,Z$ be random variables. Then
  \[
  \MI[ X,Y ; Z] = \Expect{(x,y)\sim (X,Y)}{\DKL{Z|_{X=x,Y=y}}{Z}}.
  \]
\end{fact}

\begin{fact}\label{fact:MI_chainrule}
  Let $X,Y_1,\ldots,Y_n$ be random variables. Then
  \[
  \MI[ X; Y_1,\ldots,Y_n] = \sum\limits_{i=1}^{n} \MI[X; Y_i~|~Y_{<i}].
  \]
\end{fact}

\begin{fact}\label{fact:MI_cond}
  Let $X,Y,Z$ be random variables. Then
  $\MI[ X; Y~|~Z]\leq \MI[X; Y,Z]$.
\end{fact}


\section{Probability Theory}
We will use several standard distance measures between probability measures as well as relations between
them that we give in this section.

\begin{definition}
  Let $X, Y$ be random variables supported over a finite set $\Omega$. The statistical distance between
  $X$ and $Y$ is defined as
  \[
  {\sf SD}(X, Y) = \frac{1}{2}\sum\limits_{w\in \Omega}\card{\Prob{}{X = w} - \Prob{}{Y = w}}
  \]
\end{definition}

\begin{definition}
  Let $X, Y$ be random variables supported over a finite set $\Omega$. The Hellinger distance
  between $X$ and $Y$ is defined as
  \[
  {\sf Hellinger}(X, Y) = \sqrt{\frac{1}{2}\sum\limits_{w\in \Omega}\card{\sqrt{\Prob{}{X = w}} - \sqrt{\Prob{}{Y = w}}}^2}
  \]
\end{definition}

\subsection{A Relation Between KL-divergence and Hellinger Distance}
The KL divergence metric does not obey a triangle inequality, and the following lemma replaces it in our inteded
application. It asserts that if the KL divergence between $X$ and $Z$ is small, and the KL divergence between
$Y$ and $Z$ is small, then the Hellinger distance between $X$ and $Y$ is small.
\begin{lemma}\label{lem:move_KL_hell}
  Let $X, Y$ and $Z$ be random variables supported on a finite set $\Omega$. Then
  \[
   2{\sf Hellinger}(X, Y)^2\leq \DKL{Z}{X} + \DKL{Z}{Y}.
  \]
\end{lemma}
\begin{proof}
  Let us think of $X$ and $Y$ as fixed random variables, and attempt to minimize
  $\DKL{Z}{X} + \DKL{Z}{Y}$ over all random variables $Z$. In other words, we wish find
  non-negative numbers $(p_Z(w))_{w\in\Omega}$ summing up to $1$ minimizing the form
  \[
  \sum\limits_{w\in \Omega}{p_Z(w)\log\left(\frac{p_Z(w)}{p_X(w)}\right) + p_Z(w)\log\left(\frac{p_Z(w)}{p_Y(w)}\right)}
  =\sum\limits_{w\in \Omega}{p_Z(w)\log\left(\frac{p_Z(w)^2}{p_X(w)p_Y(w)}\right)}
  \]
  Using Lagrange multipliers, we get that there exists $\lambda\in\mathbb{R}$ such that the optimum satisfies the equations
  \[
  \log\left(\frac{p_Z(w)^2}{p_X(w)p_Y(w)}\right)-\frac{1}{\ln 2} - \lambda = 0\qquad\qquad\forall w\in\Omega.
  \]
  In other words, $p_Z(W) = c\sqrt{p_X(w)p_Y(w)}$ for some $c>0$, and we next compute the constant $c$. By the constraint that the
  sum is $1$, we get
  \[
  1
  = \frac{c}{2}\sum\limits_{w\in \Omega}2\sqrt{p_X(w)p_Y(w)}
  = \frac{c}{2}(2-\sum\limits_{w\in \Omega} (\sqrt{p_X(w)}-\sqrt{p_Y(w)})^2)
  =c(1 - {\sf Hellinger}(X,Y)^2),
  \]
  so $c = \frac{1}{1 - {\sf Hellinger}(X,Y)^2}$. Plugging that into $p_Z(w)$, we get that the minimum of $\DKL{Z}{X} + \DKL{Z}{Y}$
  is at least
  \[
  \sum\limits_{w\in \Omega}{p_Z(w)\log\left(\frac{p_Z(w)^2}{p_X(w)p_Y(w)}\right)}
  =\sum\limits_{w\in \Omega}{c\sqrt{p_X(w)p_Y(w)}\log\left(\frac{c^2 p_X(w)p_Y(w)}{p_X(w)p_Y(w)}\right)}.
  \]
  Simplifying, we get that this is equal to
  \[
  c\log c^2 \sum\limits_{w\in \Omega}{\sqrt{p_X(w)p_Y(w)}}
  =2\log c
  =-2\log\left(1 - {\sf Hellinger}(X,Y)^2\right)
  \geq 2{\sf Hellinger}(X,Y)^2,
  \]
  where the last inequality uses $\log(s)\leq s-1$ for all $s>0$.
\end{proof}

\subsection{Small KL-divergence and high probability events}
The next few lemmas are concerned with a pair of distributions that have a small KL-divergence.
They show that an event has probability close to $1$ with respect to one distribution if and only
if it has probability close to $1$ with respect to the other distributions; we then generalize
this facts to bounded functions.
\begin{lemma}\label{lem:move_high_prob_events}
  Suppose $P$ and $Q$ are distributions such that $\DKL{P}{Q}\leq \eta\leq 1/100$, and suppose
  that $E$ is an event.
  \begin{enumerate}
    \item If $P[E]\geq 1-\eta$, then $Q(E)\geq 1-10\eta$.
    \item If $Q(E)\geq 1-\eta$, then $P(E)\geq 1-10\eta$.
  \end{enumerate}
\end{lemma}
\begin{proof}
  By the data processing inequality, $\eta\geq \DKL{P}{Q}\geq \DKL{{\sf Ber}(P(E))}{{\sf Ber}(Q(E))}$.

  For the first item, if we assume $Q(E)\leq 1-10\eta$ then we get that
  \begin{align*}
  \eta \geq
  \DKL{{\sf Ber}(1-\eta)}{{\sf Ber}(1-10\eta)}
  &=
  (1-\eta)\log\left(\frac{1-\eta}{1-10\eta}\right)
  +\eta\log\left(\frac{\eta}{10\eta}\right)\\
  &=(1-\eta)\log\left(1+\frac{9\eta}{1-10\eta}\right)
  -\eta\log\left(10\right)\\
  &\geq(1-\eta)\log(1+9\eta) - \eta\log(10).
  \end{align*}
  Using the fact that $\log(1+s)\geq s/2$ for $s\leq 1$, we get that the last expression is at least
  $4.5\eta(1-\eta) - \eta\log(10)$, and as $\log(10)\leq 3.4$ we get that it is at least
  $1.1\eta - 4.5\eta^2\geq 1.1\eta - 0.045\eta >\eta$, and contradiction.

  For the second item, if we assume that $P(E)\leq 1-10\eta$, then we get that
  \begin{align*}
  \eta \geq
  \DKL{{\sf Ber}(1-10\eta)}{{\sf Ber}(1-\eta)}
  &=
  (1-10\eta)\log\left(\frac{1-10\eta}{1-\eta}\right)
  +10\eta\log\left(\frac{10\eta}{\eta}\right)\\
  &=(1-10\eta)\log\left(1-\frac{9\eta}{1-\eta}\right)
  +10\log\left(10\right)\eta\\
  &\geq
  -(1-10\eta)\frac{18\eta}{1-\eta} + 10\log(10)\eta,
  \end{align*}
  where we used $\log(1-s)\geq -2s$ for $s\leq 0.1$.
  As $1/(1-\eta)\leq 1+2\eta$, we
  get that the last expression is at least
  \[
  (10\log(10) - 18)\eta - 216 \eta^2
  \geq 10\eta - \frac{216}{100}\eta^2
  >\eta,
  \]
  and contradiction.
\end{proof}

\begin{lemma}\label{lem:trivial_move_expect}
  Suppose $X, Y$ are random variables in $[0,1]$, and $\DKL{X}{Y}\leq \eta$,
  for $\eta\leq 1/100$. Then
  \begin{enumerate}
    \item If $\Expect{}{X}\geq 1-\eta$, then $\Expect{}{Y}\geq 1-10\eta$.
    \item If $\Expect{}{Y}\geq 1-\eta$, then $\Expect{}{X}\geq 1-10 \eta$.
  \end{enumerate}
\end{lemma}
\begin{proof}
  We prove the first item, and the second item follows similarly.

  Consider the following randomized process $\mathcal{T}$: given a number $x\in [0,1]$
  sample a Bernoulli random variable $b$ such that $\Prob{}{b=1} = x$. Then by the data processing
  inequality we have that
  \[
    \eta\geq \DKL{X}{Y}\geq \DKL{\mathcal{X}}{\mathcal{Y}}.
  \]
  Let $E$ be the event that the output of the process is $1$. Then
  \[
    \Prob{}{\mathcal{X}\in E} = \Expect{}{X}\geq 1-\eta,
  \]
  so we get by Lemma~\ref{lem:move_high_prob_events} that
  \[
    \Expect{}{Y} = \Prob{}{\mathcal{Y}\in E} \geq 1-10\eta.\qedhere
  \]
\end{proof}

\begin{lemma}\label{lem:trivial_move_expect2}
  Suppose $f\colon \Omega\to[0,2]$, $\eta\leq 1/100$ and $P,Q$ are distributions such that
  $\DKL{P}{Q}\leq \eta$. Then
  \begin{enumerate}
    \item If $\Expect{z\sim P}{f(z)}\leq \eta$, then $\Expect{z\sim Q}{f(z)}\leq 20\eta$.
    \item If $\Expect{z\sim Q}{f(z)}\leq \eta$, then $\Expect{z\sim P}{f(z)}\leq 20\eta$.
  \end{enumerate}
\end{lemma}
\begin{proof}
  We prove the first item, and the second item follows similarly.

  Consider the random variables $X = 1-\frac{1}{2}f(z)$ where $z\sim P$, and
  $Y = 1-\frac{1}{2}f(z)$ where $z\sim Q$. By the Data processing inequality
  $\DKL{X}{Y}\leq \DKL{P}{Q}\leq \eta$, and by assumption $\Expect{}{X}\geq 1-\eta$,
  so by Lemma~\ref{lem:trivial_move_expect} we get that $0.5\Expect{z\sim Q}{f(z)} = 1-\Expect{}{Y}\leq 10\eta$.
\end{proof}

\subsection{Correlated sampling}
Correlated sampling is an important motive in parallel repetition theorems that has been introduced in~\cite{Holenstein}.
Below, we state a version of it for pseudo-distributions, and include a proof for completeness.
\begin{lemma}\label{lemma:correlated sampling}
  There exists a randomized procedure $\mathcal{R}$.

  Let $\eta\leq 1/100$ and $0<\eps<1/2$. Suppose $P$ and $Q$ are distributions, and $\tilde{P}, \tilde{Q}$ are pseudo-distributions
  all over $\Omega$ such that
  \begin{enumerate}
    \item ${\sf SD}(\tilde{P},\tilde{Q})\leq \eta$
    \item for all $w\in \Omega$, $(1-\eps)P(w)\leq \tilde{P}(w)\leq (1+\eps)P(w)$, $(1-\eps)Q(w)\leq \tilde{Q}(w)\leq (1+\eps)Q(w)$.
  \end{enumerate}
  Then, Alice generates a sample $p\sim \tilde{P}$, Bob generates a sample $q\sim \tilde{Q}$ such that
  \[
  \Prob{}{p\neq q}\leq O(\xi).
  \]
\end{lemma}
\begin{proof}
  Alice and Bob consider, as public randomness, an infinite string of uniform tuples from $\Omega\times [0,1+\eps]$, say
  $(w_1,r_1),(w_2,r_2)$ and so on. Alice picks the first $i$ such that
  \[
  r_i\leq \tilde{P}(w_i),
  \]
  and outputs $w_i$. Similarly Bob picks the first $j$ such that
  \[
  r_j\leq \tilde{Q}(w_j),
  \]
  and outputs $w_j$. It is clear that $w_i\sim \tilde{P}$ and $w_j\sim \tilde{Q}$, and we next
  bound the probability that $i\neq j$. This happens if for the smallest $k$ such that
  $r_k\leq \max(\tilde{P}(w_k),\tilde{Q}(w_k))$, it holds that
  $r_k > \min(\tilde{P}(w_k),\tilde{Q}(w_k))$. The probability for that is
  \begin{align*}
  &\cProb{(r,w)}{r\leq \max(\tilde{P}(w),\tilde{Q}(w))}{r > \min(\tilde{P}(w),\tilde{Q}(w))}\\
  &\qquad=\frac{\sum\limits_{w}\Prob{r}{\min(\tilde{P}(w),\tilde{Q}(w))< r < \max(\tilde{P}(w),\tilde{Q}(w))}}{\sum\limits_{w} \Prob{r}{r\leq \max(\tilde{P}(w),\tilde{Q}(w))}}.
  \end{align*}
  We finish the proof by upper bounding the numerator by $\eta$, and lower bounding the denominator by $1/4$.
  The numerator may be bounded as
  \[
  \sum\limits_{w}\card{\max(\tilde{P}(w),\tilde{Q}(w)) - \min(\tilde{P}(w),\tilde{Q}(w))}
  =\sum\limits_{w}\card{\tilde{P}(w)-\tilde{Q}(w)}
  \leq \eta.
  \]
  As for the denominator, it may be lower bounded by
  \[
  \sum\limits_{w} \Prob{r}{r\leq \tilde{P}(w)}
  =\sum\limits_{w} \frac{\tilde{P}(w)}{1+\eps}
  \geq
  \sum\limits_{w} \frac{(1-\eps)P(w)}{1+\eps}
  =\frac{1-\eps}{1+\eps}\geq \frac{1}{4}.\qedhere
  \]
\end{proof}

\end{document}